\documentclass[a4paper, 11pt]{article}

\def\dash---{\kern.16667em---\penalty\exhyphenpenalty\hskip.16667em\relax}

\usepackage[margin=1in]{geometry}
\usepackage{amsmath}
\usepackage{amssymb}
\usepackage{amsthm}
\usepackage{bbm}
\usepackage{hyperref}
\usepackage[ruled, vlined, linesnumbered]{algorithm2e}
\usepackage{multicol}
\usepackage[T1]{fontenc}
\usepackage{color}
\usepackage{soul}
\usepackage{paralist}
\usepackage{wrapfig}
\usepackage{authblk}

\makeatletter
\renewcommand\paragraph{\@startsection{subparagraph}{5}{\z@}%
                                      {1.2ex \@plus1ex \@minus .2ex}%
                                      {-1em}%
                                      {\normalsize\bfseries}}
\makeatother
\newtheoremstyle{mythm}%	name
     {3pt}%     Space above
     {3pt}%     Space below
     {}%        Body font
     {}%        Indent amount (empty = no indent, \parindent = para indent)
     {\bfseries}%Thm head font
     {.}%       Punctuation after thm head
     { }%	Space after thm head: " " = normal interword space; %       \newline = linebreak
     {}%        Thm head spec (can be left empty, meaning `normal')

\theoremstyle{mythm}

\newtheorem{thm}{Theorem}[section]
\newtheorem{lem}[thm]{Lemma}
\newtheorem{cor}[thm]{Corollary}
\newtheorem{prop}[thm]{Proposition}
\newtheorem{claim}[thm]{Claim}

\newtheorem{dfn}[thm]{Definition}

\newtheorem{ntn}[thm]{Notation}
\newtheorem{obs}[thm]{Observation}

\newcommand{\dft}[1]{\textbf{\textit{#1}}}

\newcommand{\abs}[1]{\left|#1\right|}
\newcommand{\norm}[1]{\left\|#1\right\|}
\newcommand{\paren}[1]{\left(#1\right)}
\newcommand{\set}[1]{\left\{#1\right\}}
\newcommand{\sucht}{\,\middle|\,}
\renewcommand{\th}{{}^{\text{th}}}

\newcommand{\N}{\mathbb{N}}
\newcommand{\R}{\mathbb{R}}

\newcommand{\FC}{\textsf{FC}}
\newcommand{\SC}{\textsf{SC}}
\newcommand{\FT}{\textsf{FT}}
\newcommand{\ST}{\textsf{ST}}
\newcommand{\FCone}{\textsf{FC-1}}
\newcommand{\SCone}{\textsf{SC-1}}
\newcommand{\FTone}{\textsf{FT-1}}
\newcommand{\STone}{\textsf{ST-1}}
\newcommand{\FCtwo}{\textsf{FC-2}}
\newcommand{\SCtwo}{\textsf{SC-2}}
\newcommand{\FTtwo}{\textsf{FT-2}}
\newcommand{\STtwo}{\textsf{ST-2}}

\newcommand{\e}{\varepsilon}

\newcommand{\bfp}{\mathbf{p}}
\newcommand{\calC}{\mathcal{C}} % set of clusters
\newcommand{\calE}{\mathcal{E}} % edges in cluster graph
\newcommand{\calG}{\mathcal{G}} % cluster graph
\newcommand{\calO}{\mathcal{O}} % O-notation

\newcommand{\hnom}{h^{\mathrm{nom}}} % the rate of a node's nominal clock
\newcommand{\indf}{\gamma}

\newcommand{\mx}{\mathrm{max}}

\newcommand{\tilL}{\widetilde{L}}
\renewcommand{\bar}[1]{\overline{#1}}

\newcommand{\thmax}{\vartheta_{\mathrm{max}}} %% the largest value of theta imaginable!
\newcommand{\thg}{\vartheta_g}              %% generic value of theta
\newcommand{\Err}{\mathfrak{E}}      %% The crudest upper bound on the pulse diameter
\newcommand{\essf}{e_{f}^\infty}      %% steady state error for fast clusters     
\newcommand{\essg}{e_{g}^\infty}      %% steady state error for general clusters
\newcommand{\esss}{e_{s}^\infty}      %% steady state error for slow clusters

\DeclareMathOperator{\ClusterSync}{ClusterSync}
\DeclareMathOperator{\InterclusterSync}{InterclusterSync}

\title{Fault Tolerant Gradient Clock Synchronization}

\author{Johannes Bund}
\author{Christoph Lenzen}
\author{Will Rosenbaum}
\affil{Max Planck Institute for Informatics\\Saarbr\"{u}cken, Germany}

\date{}

\begin{document}

\maketitle

\begin{abstract}
Synchronizing clocks in distributed systems is well-understood, both in terms of fault-tolerance in fully connected systems and the dependence of local and global worst-case skews (i.e., maximum clock difference between neighbors and arbitrary pairs of nodes, respectively) on the diameter of fault-free systems. However, so far nothing non-trivial is known about the local skew that can be achieved in topologies that are not fully connected even under a single Byzantine fault. Put simply, in this work we show that the most powerful known techniques for fault-tolerant and gradient clock synchronization are compatible, in the sense that the best of both worlds can be achieved simultaneously.

Concretely, we combine the Lynch-Welch algorithm~\cite{welch88new} for synchronizing a clique of $n$ nodes despite up to $f<n/3$ Byzantine faults with the gradient clock synchronization (GCS) algorithm by Lenzen et al.~\cite{lenzen10tight} in order to render the latter resilient to faults. As this is not possible on general graphs, we augment an input graph $\mathcal{G}$ by replacing each node by $3f+1$ fully connected copies, which execute an instance of the Lynch-Welch algorithm. We then interpret these clusters as supernodes executing the GCS algorithm, where for each cluster its correct nodes' Lynch-Welch clocks provide estimates of the logical clock of the supernode in the GCS algorithm. By connecting clusters corresponding to neighbors in $\mathcal{G}$ in a fully bipartite manner, supernodes can inform each other about (estimates of) their logical clock values. This way, we achieve asymptotically optimal local skew, granted that no cluster contains more than $f$ faulty nodes, at factor $O(f)$ and $O(f^2)$ overheads in terms of nodes and edges, respectively. Note that tolerating $f$ faulty neighbors trivially requires degree larger than $f$, so this is asymptotically optimal as well.
\end{abstract}

\section{Introduction and Related Work}\label{sec:intro}

Synchronizing clocks across distributed systems is a fundamental task. It may be used for coordination, e.g.\ in a time division multiple access scheme to a shared resource such as a wireless channel, play a crucial role in (distributed) measurements by enabling to correctly correlate data, or be the basis of a decentralized system clock for a System-on-Chip or Network-on-Chip.

Gradient clock synchronization (GCS) algorithms aim to minimize the \emph{local skew,} i.e., the worst-case phase difference between the logical clocks computed by neighbors in the network graph. While the \emph{global skew}\dash---the worst-case skew between any pair of clocks in the network\dash---is linear in the network diameter $D$~\cite{biaz01closed},\footnote{Throughout this paper, we assume that all parameters like delays and delay uncertainties on links are uniform, meaning that this refers to the hop diameter of the network. However, prior results generalize to heterogeneous settings in a natural way. For instance, in this particular case the hop diameter needs to be replaced by the weighted diameter of the graph in which each link has weight equal to its delay uncertainty.} a tight bound of $\Theta(\log_b D)$ on the local skew has been established~\cite{lenzen10tight}. Here, the base $b = \mu/\rho$, subject to the constraints that (i) the nodes' hardware clocks always run at rates between $1$ and $1+\rho$, (ii) nodes' logical clocks always run at rates between $1$ and $(1+\rho)(1+\mu)$, and (iii) $\rho<\mu \in \calO(1)$.\footnote{In~\cite{lenzen10tight}, the base is shown to be $\Theta(\mu/\rho)$ for a $\mu$ that exceeds $\rho$ by at least some constant factor. The stronger result follows from the more refined analysis in~\cite{kuhn18gradient}.} 

In other words, the problem of minimizing the local skew is well-studied and, even under worst-case assumptions, small bounds can be achieved despite logical clocks that behave like slightly worse hardware clocks in terms of their rates. Results in the same vein hold for dynamic graphs~\cite{kuhn10gradient,kuhn18gradient}. This extends the above to settings with crash faults, as from the perspective of the remaining system, crashing a node is equivalent to removing all of its incident links. On top of this, the algorithm achieving all of this is inherently  self-stabilizing, provided one ensures that excessive global skews are detected and reduced by, e.g., a reset procedure.

So what is this paper about? The GCS algorithm utterly fails in face of non-benign faults. Even a single node that refuses to adjust its logical clock rate would invalidate any non-trivial bound on skews. Given that distributed systems of sufficient size invariably tend to violate (overly optimistic) specifications~\cite{fallacies}, this raises the following question.
\begin{quote}
\centering
\emph{Can small local skew be achieved despite Byzantine faults?}
\end{quote}
When posing this question, obviously we need to restrict to considering skews between non-faulty nodes only. Nonetheless, the answer is trivially ``no'' in general: a node with exactly two neighbors, one of which is Byzantine, cannot reliably decide which neighbor's clock it should follow. As a Byzantine node can, e.g., run its hardware clock at slightly sub-nominal speed without a correct node with (slow) hardware clock being able to prove this, no non-trivial skew bound can be guaranteed in this scenario. More generally, a node with up to $f$ faulty neighbors must have at least $2f+1\in \Omega(f)$ neighbors to avoid trivial impossibility of synchronization. A more careful argument shows that, without cryptographical assumptions, $n>3f$ is needed even if the network is a clique~\cite{dolev84possibility,lecture18fault}. This bound is matched by the Lynch-Welch algorithm~\cite{welch88new}, which simultaneously achieves asymptotically optimal skew\dash---both globally and locally, as $D=1$ in a clique.

The Lynch-Welch algorithm, like other fault-tolerant clock synchronization algorithms designed for cliques, can be extended to networks of larger diameter in a straightforward way.\footnote{We are not aware of this being explicitly shown, but~\cite{dolev16synchronous} discusses the approach in a synchronous setting and it transfers without issue.} The idea is to set up a central cluster running the algorithm and synchronize ``slave'' clusters of nodes to it, i.e., let them ``echo'' the clock pulses generated by the central cluster without participating in the algorithm. We then can slave further clusters to the previous ones, and so on, resulting in a tree structure with each tree node corresponding to a cluster of nodes. So long as in each cluster no more than one third of the nodes is faulty, subsequent groups can synchronize both to their ``masters'' (i.e., the parent cluster in the tree) and among themselves, using the same technique based on approximate agreement~\cite{dolev86reaching} that lies at the heart of the basic algorithm.

Note that if clusters have uniform size of $3f+1$ and nodes fail independently with probability $p$, then the probability that more than $f$ nodes in a cluster are faulty is
\begin{equation}\label{eq:faulty_prob}
\sum_{i=f+1}^{3f+1} \binom{3f+1}{i} p^i (1-p)^{3f+1-i} \leq \binom{3f+1}{f+1} p^{f+1}\leq (3ep)^{f+1}\,.
\end{equation}
Thus, if the \emph{distribution} of faults across the system is benign, even small choices of $f$ can improve reliability dramatically, without causing impractically large degrees. For $f\in \Theta(\log n)$, the system as a whole operates correctly with high probability even for a constant failure probability $p$ of individual nodes.

This simplistic approach succeeds in the sense that it achieves asymptotically optimal global skew in a sparse network (assuming that at most $f$ nodes fail in each cluster). However, it does not offer a non-trivial bound on the local skew. Setting $f=0$, the algorithm specializes to a simple master-slave synchronization algorithm on (fault-free) tree topologies. If a clock pulse propagates through a line network with the global skew equally distributed over the line, this will ``compress'' the full global skew onto a single edge, cf.~\cite{locher2006oblivious}. In contrast, GCS algorithms need to take into account whether neighbors are lagging behind when deciding how to adjust their logical clocks. More sophisticated strategies are needed for a fault-tolerant GCS algorithm!

%###
\paragraph{Our contribution.}
%###
We present a simple and general transformation that takes an arbitrary network $\mathcal{G}$ and yields a slightly larger one, on which we can achieve fault-tolerant GCS with asymptotically optimal local skew. The basic idea is simple enough: essentially, the transformation is as described above, where each node is replaced by a (fully connected) cluster of $3f+1$ nodes and each edge is replaced by a complete bipartite graph between the respective clusters. We then use the Lynch-Welch algorithm~\cite{welch88new} to synchronize within clusters, and simulate the (non-fault-tolerant) GCS algorithm from~\cite{lenzen10tight} on virtual clocks defined for the clusters.

What would be a daunting task when performed from scratch is greatly simplified by employing the two algorithms (almost) as black boxes. This is made possible by leveraging the worst-case assumptions on clock rates the Lynch-Welch algorithm can handle: By pretending that the speed adjustments made by the concurrently running GCS algorithm running are simply changes in ``hardware'' clock speeds, we can fully exploit the analysis of the algorithm without having to prove statements again from scratch. Similarly, the GCS algorithm from~\cite{lenzen10tight} can be phrased such that it takes its decisions based solely on estimates of real time differences between events (which any node gets from its hardware clock) and estimates of differences between logical clocks. We exploit the latter by having nodes use their own logical clocks in the Lynch-Welch algorithm as stand-in for the (virtual) cluster clocks, which they can never know precisely.

The main obstacle to this approach is that using both algorithms as black boxes results in the problem that the ``additional'' (amortized) clock drift induced by the Lynch-Welch algorithm's corrections to nodes' clocks is proportional to the (intra-cluster) synchronization quality divided by the resynchronization interval, while the synchronization quality is proportional to the difference between maximum and minimum clock rate times the resynchronization interval. This means that a naive analysis would yield that the clock drift the GCS algorithm needs to combat is at least as large as the increase in clock speed the GCS algorithm is willing to use to do so. In other words, fast-running clocks would need to be able to outrun other fast-running clocks, leaving the GCS algorithm with no way of reducing skews at all!

We resolve this issue by exploiting the convergence properties of the Lynch-Welch algorithm in combination with the flexibility of the GCS algorithm. The GCS algorithm allows for slack between conditions under which logical clocks \emph{must} run fast or slow, respectively, and so-called \emph{triggers} that indicate when an algorithm can \emph{decide} that the clock should run fast or slow, respectively, without risking conflicts in this decision. This slack enables an implementation despite the fact that clock values of neighbors are never known exactly. However, we can also use this slack to buy some time for responding to critical skews. During this time, all (correct) nodes of a cluster will satisfy the respective trigger and unanimously run fast or slow, respectively. This means that the Lynch-Welch algorithm does not have to deal with an artificially increased clock drift due disagreement on logical clock rates within the cluster during this time, resulting in convergence to a smaller intra-cluster skew that necessitates smaller clock corrections; this is, in fact, shown by applying the existing analysis of the algorithm to this particular scenario. From this insight, we can then infer that there is a sufficient gap between the speed of clusters that currently need to be fast and those that need to be slow for the analysis of the GCS algorithm, applied as a blackbox, to prove an asymptotically optimal local skew.
\begin{thm}
  \label{thm:main}
  Let $\calG = (\calC, \calE)$ be an arbitrary network, and $G = (V, E)$ the augmented graph formed by replacing each node in $\calG$ by a clique of $k \geq 3 f + 1$ nodes and fully connecting such cliques if they correspond neighbors in $G$. Suppose messages in $G$ are subject to maximum delay $d$, delay uncertainty $U$, and hardware clock drifts are at most $\rho$. Suppose further that for each node in $\calG$, at most $f$ of the duplicated nodes are faulty. Then there exists an algorithm that computes logical clocks $L_v(t)$ for each $v \in V$ such that (i) for each $v\in V$, $L_v$ increases at rates between $1$ and $1+O(\rho)$ and (ii) for all $\set{v, w} \in E$ and $t \in \R^+$, we have
  \[
  \abs{L_v(t) - L_w(t)} = O((\rho \cdot d + U) \log D)\,,
  \]
  where $D$ is the network diameter.
\end{thm}
% CL: The following may not really buy us further cookie points. I feel that the respective discussion in the related work is more effective - pluse, it saves both work and space.
% We round off our presentation by discussing that this compatibility of techniques extends further. We sketch how the algorithm can be made self-stabilizing, by replacing Lynch-Welch by a self-stabilizing variant~\cite{khanchandani18self} and utilizing that, up to the need for a reset in case of excessive global skew, the GCS algorithm is already self-stabilizing (cf.~\cite{kuhn18gradient}). Moreover, we explain how to dynamically adjust the topology of the underlying graph $G$ using consensus.

%###
\paragraph{Organization of this paper.}
%###
Due to space restrictions, a broader account of related work on the history of fault-tolerant and gradient clock synchronization is deferred to Appendix~\ref{app:further}. Section~\ref{sec:model} discusses the model, notation, and how our algorithms adjust their clocks. We then proceed to presenting the cluster synchronization algorithm. As it is, at its core, a variant of the well-known Lynch-Welch algorithm~\cite{welch88new} that amortizes clock corrections to bound clock rates, we focus on the main statements and differences to prior work; the analysis, which is based on technical results from~\cite{khanchandani18self}, is provided in Appendix~\ref{app:lw}. We then proceed to the intercluster synchronization algorithm in Section~\ref{sec:intercluster-alg}. This section briefly discusses the key requirements of the GCS algorithm by Lenzen et al.~\cite{lenzen10tight}, and then proceeds to show how to simulate executions of this algorithm on $\calG$. Some details of proving this simulation relation, which mainly consist of calculations for finding suitable parameter choices and verifying their feasibility, are deferred to Appendix~\ref{app:gcs}.

\section{Computational Model}
\label{sec:model}

%###
\paragraph{Network.}
%###
Let $\calG = (\calC, \calE)$ be an arbitrary graph. We consider a network $G = (V, E)$ constructed in from $\calG$ in the following way. We identify each $C \in \calC$ with a set of $k$ nodes $C = \set{v_1, v_2, \ldots, v_k}$. We refer to the sets $C \in \calC$ as \dft{clusters}. For distinct clusters $B, C \in \calC$, the corresponding sets of nodes in $V$ are disjoint, i.e., $V = \bigcup_{C \in \calC} C$. The edge set $E$ contains two different ``types'' of edges defined as follows:
\begin{compactdesc}
\item[cluster edges:] for each $C \in \calC$ and $v, w \in C$, we have $(v, w) \in E$;
\item[intercluster edges:] for each $(B, C) \in \calE$, $v \in B$ and $w \in C$ we have $(v, w) \in E$.
\end{compactdesc}
% Thus, each cluster $C \in \calC$ is a clique in $G$, and for each pair of adjacent clusters $(B, C) \in \mathcal{E}$, $E$ contains all possible edges between $B$ and $C$. We refer to $G$ as the \dft{physical network}.
We assume that each vertex $v \in V$ knows the identities of its neighbors, as well as the identity of the cluster to which each neighbor belongs.
%###
\paragraph{Communication and computation.}
%###
Nodes in the network communicate by sending content-less messages, known as \emph{pulses}, to their neighbors. When a (correct) node broadcasts a pulse, all of its neighbors receive the pulse after some delay, which is itself subject to some uncertainty. We denote the maximal message delay by $d$, and the uncertainty by $U$. Thus a pulse sent by $v$ at Newtonian time $p_v$ is received by each of $v$'s neighbors at a time $t\in[p_v + d - U, p_v + d]$.

%###
\paragraph{Timing and clocks.}
%###
We assume that the entire network inhabits an inertial reference frame, and the parameter $t \in \R$ denotes an absolute, Newtonian time. The network operates in a semi-synchronous model, where each node $v \in V$ has an associated hardware clock. Hardware clocks are prone to some uncertainty, which we model by a variable rate of $v$'s clock that may change over time. Formally, for each $v$ there is a locally integrable function $h_v : \R \to \R$ satisfying $1 \leq h_v(t) \leq 1 + \rho$ for all $t \in \R$. Here $\rho > 0$ is (an upper bound on) $v$'s \dft{hardware clock drift}.\footnote{Our notation differs from some previous work in that our clock drift is always positive---hardware clocks \emph{always} run fast. This is done without loss of generality, as one can rescale time to model drifts that are both positive and negative. The benefit is to significantly reduces notational clutter in our analysis.} As in our algorithm nodes use their hardware clocks exclusively to measure time differences, we may w.l.o.g.\ assume that the hardware clocks are initialized to $0$ at time $0$. Accordingly, the hardware clock of $v$ at time $t$ is defined by $H_v(t)=\int_0^t h_v(\tau)\,d\tau$. In contrast to the reference time $t$, $v$ has access to $H_v(t)$, enabling it to approximate the time that passed between local events.
% that occured at times $t$ and $t'>t$, respectively, by evaluating $H_v(t')-H_v(t)=\int_t^{t'}h_v(\tau)\,d\tau$.

While $h_v$ determines the (unknown) rate of $v$'s hardware clock, our algorithm controls its logical clock by adjusting its rate relative to the hardware clock. Specifically, the algorithm controls two parameters: $\delta_v(t) \in \R_{\geq 0}$, and $\indf_v(t) \in \set{0, 1}$. These parameters are determined by the algorithm throughout an execution in order to maintain synchronization within (Section~\ref{sec:cluster-alg}) and between clusters (Section~\ref{sec:intercluster-alg}), respectively. The logical clock value $L_v(t)$ is computed to be
\begin{equation}
  \label{eqn:logical-clock}
  L_v(t) = \int_0^t (1 + \varphi \cdot \delta_v(\tau)) (1 + \mu \cdot \indf_v(\tau)) h_v(\tau)\, d\tau.
\end{equation}
The parameters $\varphi$ and $\mu$ are constants whose values will be determined later on, where we already fix that $0<\varphi<1$ and $\mu>0$.
% The parameters $\delta_v(t)$ and $\indf_v(t)$ are controlled by two algorithms independently operating at each node.

% The first parameter $\delta_v(t)$ is adjusted to maintain synchronization between nodes within a fixed cluster, as described in Section~\ref{sec:cluster-alg}. The second parameter, $\indf_v(t)$ adjusts the clock to maintain synchronization between clusters, as described in Section~\ref{sec:intercluster-alg}.

% Given Newtonian times $t < t'$, we define the \dft{logical duration} or \dft{logical length} of the interval $[t, t']$ for $v$ to be
% \begin{equation}
%   \label{eqn:logical-length}
%   L_v(t') - L_v(t) = \int_t^{t'} (1 + \varphi \cdot \delta_v(\tau)) (1 + \mu \cdot \indf_v(\tau)) h_v(\tau)\, d\tau.
% \end{equation}

%% W: This is no longer correct because $\delta_v(t)$ could be larger than two...
%% From~(\ref{eqn:logical-clock}) and the assumption that $1 \leq h_v(t) \leq 1 + \rho$ for all $t$, we find that for all $v \in V$ and $s, t \in \R$,
%% \begin{equation}
%%   \label{eqn:logical-clock-rate}
%%   (t - s) \leq L_v(t) - L_v(s) \leq (1 + 2 \varphi) (1 + \mu) (1 + \rho) (t - s).
%% \end{equation}

It is convenient to assume that $L_v$ is differentiable in our analyis, even though $\delta_v$, $\indf_v$, and $h_v$ may be discontinuous. However, as $L_v$ is Lipschitz continuous, it can be approximated arbitrarily well by differentiable functions for which the derivative satisfies the respective bounds, so this assumption is without loss of generality.

\paragraph{Faults.}
%###
We assume that the network $G$ contains a fixed subset $F \subseteq V$ of \dft{faulty} processes. The nodes $v \in F$ are fully Byzantine: we make no assumptions whatsoever about their behavior; in particular, they are not required to communicate by broadcast.
% While correct nodes communicate by broadcast (a single pulse is sent simultaneously to all neighbors) faulty processes may send unicast messages (i.e., pulses are sent to different neighbors at different times).
We assume that the number of faulty nodes \emph{within each cluster} is bounded by parameter $f$, i.e., $f \geq \max_{v \in C} \abs{F \cap C}$, and require that $k \geq 3 f + 1$. Recall that Inequality~\eqref{eq:faulty_prob} relates this deterministic requirement to a setting in which nodes fail uniformly and independently at random with probability $p$; in this case, one can tolerate a value of $p$ up to roughly $n^{1/f}$. In contrast, an adversarial placement of faults would necessitate degrees larger than the total number of faults.

% CL: not here; also, see comment in intro and the discussion in the related work section
% [Say something about transient faults, self stabilization]

%###
\paragraph{Initialization.}
%###
Our analysis requires that at time $0$ none of the invariants (read: skew bounds) of our algorithm are violated. To achieve this, the following two options for initializing the system come to mind.
\begin{compactitem}
  \item We perform an initial flooding to trigger node initialization, where we (formally) extend the hardware and logical clock functions of nodes to include the interval between the starting time of the flooding and their initialization in a manner consistent with the model. To maintain tight bounds on the initial skew in a fault-tolerant way, this can e.g.\ be done as described in the introduction (without fixing communication to a tree). Note that, depending on the precise model assumptions, the source cluster(s) of such a flooding operation may need to perform an ``internal'' approximate agreement step to synchronize their outgoing messages to other clusters, or even run consensus to determine when to initiate the procedure.
  \item If the above method is not applicable, after initialization we can first wait until the clusters have stabilized to small internal skews by running the Lynch-Welch algorithm based on (an arbitrarily loose, but known) bound on their internal skew at initialization~(cf.~\cite{khanchandani18self}). We then treat all inter-cluster edges as newly inserted in the dynamic graph model considered in~\cite{kuhn18gradient}, which ensures stabilization to optimal skews within $O(\mathcal{S}/\mu)$ time assuming that the global skew is bounded by $\mathcal{S}$. To ensure fault-tolerance, adjacent clusters run consensus on when to add an edge $e\in \mathcal{E}$, i.e., we avoid inconsistent opinions among correct nodes whether an edge is currently ``considered'' by the algorithm or not.
\end{compactitem}
Neither of the above solutions offers significant novelty, but the involved technicalities would significantly complicate the description and analysis of our algorithm. Accordingly, in the interest of simplicity and readability, we assume that all nodes simultaneously wake up and initialize their logical clocks at time $0$.

%% \section{Algorithm Description}

%% Our algorithm can be viewed as a hybrid of two well-known clock synchronization algorithms: Lynch-Welsh~[cite] and the gradient clock synchronization (GCS) algorithm of Lenzen, Locher, and Wattenhofer~[cite]. Within each cluster $K_v$, we apply a variant of the Lynch-Welsh clock synchronization algorithm. This algorithm ensures that all nodes within the cluster have closely synchronized clocks. Thus we can define a cluster clock $C_v(t)$ such that all logical clocks $L_v(t)$ are close to $C_v(t)$.

%% Synchronization between nodes in adjacent clusters is achieved by applying the GCS algorithm. For each node $v$ and adjacent cluster $K_w$, $v$ maintains an estimate of the cluster clock $C_w(t)$. Using these estimates, $v$ determines if it should be in fast or slow mode as specified by GCS algorithm. That is, $v$ simulates the behavior of an instance of the GCS algorithm where each neighboring cluster $K_w$ is replaced with a single node $w$ whose logical clock is (approximately) $C_w(t)$. The node $v$ then acts by adjusting the hardware clock rate accordingly (i.e., setting $\indf_v(\tau)$ in Equation~(\ref{eqn:hardware-clock})).

\section{Cluster Algorithm}
\label{sec:cluster-alg}

In this section, we only consider nodes and edges within a fixed cluster $C \in \calC$. Within $C$, the logical clocks maintain synchronization by using a variant of the Lynch-Welsh algorithm~\cite{welch88new} very similar to the one described by Khanchandani and Lenzen~\cite{khanchandani18self}. Following the description in~\cite{khanchandani18self}, the cluster algorithm proceeds in rounds. During each round, each (correct) node pulses once at a prespecified logical time. Once a node records the relative times of its neighbors' pulses, it computes an adjustment to its logical clock using an ``approximate agreement'' step as in the Lynch-Welsh algorithm (cf.~\cite{dolev86reaching}). The clock adjustment is then made by setting $\delta_v(t)$ in Equation~\ref{eqn:logical-clock} to an appropriate (non-negative) value for the remainder of the round. The length of each round is inductively defined, where the initial round's length depends on (an upper bound on) the initial clock skew in the cluster.

%###
\paragraph{Algorithm description.}
%###
Each round $r\in \N$ consists of three phases, of logical durations $\tau_1(r)$, $\tau_2(r)$, and $\tau_3(r)$, respectively; the total round length is $T(r) = \tau_1(r) + \tau_2(r) + \tau_3(r)$. The phases play the following roles.
\begin{compactdesc}
\item[Phase 1.] This phase is sufficiently long for all nodes within a cluster to have transitioned to round $r$ by the end of phase~1. Each node sends a pulse at the end of its phase~1.
\item[Phase 2.] During this phase, each node waits to receive pulses from its cluster neighbors. Phase~2 is sufficiently long that by the end of this phase, each node will have received pulses from all of its (correct) neighbors within its cluster. At the end of the phase, each node $v$ computes an adjustment $\Delta_v(r)$ to its own logical clock.
\item[Phase 3.] During this final phase, $v$ implements the clock adjustment computed at the end of phase~2 by setting $\delta_v(t)$ to an appropriate value for the duration of the phase.
\end{compactdesc}
During phases~1 and~2, each node simply sets $\delta_v(t) = 1$.
% This value is only changed for the duration of phase~3 in order to achieve the desired correction to $v$'s logical clock.
We give pseudo-code for $\ClusterSync$ in Algorithm~\ref{alg:cluster-sync}. The algorithm uses the following notation. For each round $r$ and nodes $v, w \in C$:
\begin{compactitem}
\item $t_v(r)$ is the Newtonian time at which $v$ begins round $r$;
\item $p_v(r)$ is the Newtonian time at which $v$ sends its pulse in round $r$;
\item $t_{w v}(r)$ is the Newtonian time at which $v$ receives the pulse $w$ sent in ($w$'s) round $r$.
\end{compactitem}
In cases where $v$ or $w$ is faulty, the values above may not be well-defined. In such cases, we can assign their values arbitrarily. While the notation above is convenient for describing the algorithm, we emphasize that nodes cannot access the values $t_v(r)$, etc., directly. Instead, $v$ stores, for example, $L_v(t_{w v}(r))$, the logical time at which it receives $w$'s round $r$ pulse.
\SetKwProg{AtTime}{at-time}{ do}{end}
\begin{algorithm}
  \caption{$\ClusterSync(v, \tau_1, \tau_2, \tau_3)$}\label{alg:cluster-sync}
  $L_v\leftarrow 0$\;
  \ForEach{round $r \in \N$}{
    $\delta_v \leftarrow 1$\;
    start listening for messages\;
    \AtTime{$L_v(t_v(r)) + \tau_1(r)$\label{cs:phase-1}}{
      broadcast clock pulse\;
    }
    \AtTime{$L_v(t_v(r)) + \tau_1(r) + \tau_2(r)$\label{cs:phase-2}}{
      $S_v\leftarrow \emptyset$\tcp*{multiset, ordered ascendingly}
      \ForEach{node $w \in C_v$}{
        $\tau_{w v} \leftarrow L_v(t_{w v}) - L_v(t_{vv})$\;
        $S_v\leftarrow S_v\cup \{\tau_{wv}\}$\;
      }
      $\Delta_v(r) \leftarrow \frac{S_v^{f + 1} + S_v^{n-f}}{2}$\label{cs:set-Delta}\tcp*{$S_v^i$ is the $i$-th element of $S_v$}
      $\delta_v \leftarrow 1 - \paren{1 + \frac{1}{\varphi}} \frac{\Delta_v(r)}{\tau_3(r)+\Delta_v(r)}$\label{cs:set-delta}\;
    }
    \AtTime{$L_v(t_v(r)) + \tau_1(r) + \tau_2(r) + \tau_3(r)$\label{cs:phase-3}}{
      end round $r$\;
    }
  }
\end{algorithm}

We define the \dft{nominal rate} of node $v$ as
\begin{equation}
  \label{eqn:nominal-rate}
  \hnom_v(t) = (1 + \varphi) (1 + \mu \cdot \indf_v(t)) h_v(t)\,;
\end{equation}
it can be thought of the ``hardware'' clock rate that our variant of the Lynch-Welch algorithm needs to handle. We first show that amortizing the clock adjustment in phase 3 as described results in the intended clock adjustment of $-\Delta_v(r)$.
% In phases~1 and~2, $v$ sets $\delta_v(t) = 1$, so that the rate of $v$'s logical clock is equal to its nominal clock rate. The following lemma shows that setting $\delta_v(t)$ in Line~\ref{cs:set-delta} ensures that the \emph{nominal} length of round $r$ is $T(r) + \Delta_v(r)$.
\begin{lem}
  \label{lem:nominal-round-length}
  Fix a non-faulty node $v \in C \setminus F$ and a round $r\in \N$. Then the nominal length of round $r$ for $v$ is $T(r) + \Delta_v(r)$. That is,
  $\int_{t_v(r)}^{t_v(r + 1)} \hnom_v(\tau)\,d\tau = T(r) + \Delta_v(r)$.
\end{lem}

Lemma~\ref{lem:nominal-round-length} shows that Algorithm~\ref{alg:cluster-sync} achieves the same effect at the end of the round as the variant of the Lynch-Welsh algorithm given in~\cite{khanchandani18self}. In~\cite{khanchandani18self}, $L_v$ is discontinuously adjusted by $-\Delta_v(r)$ at time $s_v(r)$, and increases at the nominal rate at all other times. The advantage of our formulation of Algorithm~\ref{alg:cluster-sync} is that $L_v(t)$ is continuous and increases at a rate that can be kept close to the nominal rate by choosing $\tau_3(r)$ sufficiently large. This is essential for the GCS inter-cluster algorithm and its analysis, as it requires to bound clock rates from above and below.

In Appendix~\ref{app:lw}, we repeat the analysis from~\cite{khanchandani18self} with the (minor) adjustments necessary for our variant of the algorithm. This yields the following main result. Define
\begin{align*}
\vartheta_g&=(1+\rho)(1+\mu)\\
\alpha & = \frac{6 \vartheta^2_g \varphi + 5 \vartheta_g \varphi - 9 \varphi + 2 \vartheta_g^2 - 2}{2 \varphi (\vartheta_g + 1)}\\
\beta & = \left(3 \vartheta_g - 1 + \frac{\vartheta_g-1}{\varphi}\right) U + (\vartheta_g - 1) d\\
E & = \frac{\beta}{1-\alpha}
\end{align*}
We can choose parameters such that $\alpha<1$. Under this condition, this yields the following bound on skew within clusters.
\begin{cor}
  \label{cor:cluster-skew-bound}
  Suppose the preconditions of Proposition~\ref{prop:cluster-error-bound} are satisfied. With $E$ and $T=\tau_1+\tau_2+\tau_3$ as in the proposition, we have for all times $t$ and non-faulty nodes $v, w \in C \setminus F$ that
  \[
  \abs{L_v(t) - L_w(t)} \leq \vartheta_g \cdot E + (\vartheta_g - 1) T < 2\vartheta_g\cdot E\,,
  \]
  where $r$ is the largest round such that $p_v(r), p_w(r) \leq t$, and $\vartheta_g$ is defined in~(\ref{eqn:gamma}).
\end{cor}

% \begin{ntn}
%   In accordance with Corollary~\ref{cor:cluster-skew-bound}, we define the local cluster skew bound by
%   \begin{equation}
%     \label{eqn:cluster-skew-bound}
%     \calE(r) = \vartheta_g \cdot e(r) + (\vartheta_g - 1) \paren{T(r) + \tau_1(r + 1) - \tau_1(r)}.
%   \end{equation}
%   Thus, we can write the conclusion of Corollary~\ref{cor:cluster-skew-bound} as
%   \begin{equation}
%     \label{eqn:simple-cluster-skew}
%     \abs{L_v(t) - L_w(t)} \leq \calE(r).
%   \end{equation}
% \end{ntn}

%###
\paragraph{Cluster clocks and estimates.}
%###
We now derive some summarize some further results based on the Lynch-Welch technique and its analysis that we use in Section~\ref{sec:intercluster-alg} to achieve the desired simulation of the GCS algorithm from~\cite{kuhn18gradient,lenzen10tight}. We start by defining the simulated cluster clocks.
\begin{dfn}
  \label{dfn:cluster-clock}
  Fix a cluster $C$. For each time $t$ define $L_C^+(t)$ and $L_C^-(t)$ to be the maximum and minimum values (respectively) of logical clock values of non-faulty clocks in $C$ at time $t$. That is, $L_C^+(t) = \max \set{L_v(t) \sucht v \in C \setminus F}$ and $L_C^-(t) = \min \set{L_v(t) \sucht v \in C \setminus F}$. We define $C$'s \dft{cluster clock} $L_C$ by the formula
  $L_C(t) = (L_C^+(t) + L_C^-(t))/2$.
%   That is, $L_C(t)$ is the midpoint of the extremal logical clock values of non-faulty processes in $C$. 
\end{dfn}

At several points, we will appeal to the following result about cluster clocks.

\begin{obs}
  \label{obs:cluster-clock-rate}
  Suppose that over some interval $[t', t]$ the logical clock of each $v \in C \setminus F$ increases at a rate at most (at least) $\vartheta$. Then we have $L_C(t) - L_C(t') \leq \vartheta \cdot (t - t')$ ($L_C(t) - L_C(t') \geq \vartheta \cdot (t - t')$). Indeed, if each clock individually satisfies some upper (lower) bound on its rate, then in particular the rates of $L_C^+(t)$ and $L_C^-(t)$ satisfy the same bound, hence so does $L_C(t)$.
\end{obs}

Suppose $C$ is a cluster and $w$ a node adjacent to $C$ (i.e., $w \in B$ where $(B, C) \in \calE$). Then $w$ computes an estimate $\tilL_C^w(t)$ of $L_C(t)$ as follows. The node $w$ listens to the pulses of nodes in $C$ and simulates the $\ClusterSync$ algorithm, without sending pulses itself. Then $w$ takes $\tilL_C^w(t)$ to be the logical clock value computed in its simulation of $\ClusterSync$. Applying the analysis of the algorithm (unchanged!) to $w$'s estimate $\tilL_C^w(t)$, we obtain the following guarantee.

\begin{cor}
  \label{cor:cluster-clock-estimate}
  Let $C$ be a cluster and $w$ a node adjacent to $C$. Suppose $\set{e(r)}$ is as in Proposition~\ref{prop:cluster-error-bound}. Then for all $v \in C \setminus F$ and times $t$ we have $|\tilL_C^w(t) - L_v(t)| \leq \calE(r)$, where $r$ is the largest round such that $p_v(r), \tilde{p}_w(r) \leq t$.\footnote{Here $\tilde{p}_w(r)$ denotes the time that $w$ would have sent its $r\th$ pulse in its simulation of $\ClusterSync$ listening to pulses in $C$.} Thus for all $t$ we have
  $|\tilL_C^w(t) - L_C(t)| \leq \calE(r)/2$.
\end{cor}

%###
\paragraph{Bounds for unanimous clusters.}
%###
In this section we bounds on the amortized rates of clusters clocks when all correct nodes are running in fast or slow modes. We say that a cluster $C$ is \dft{unanimous at time $t$} if either (1) for all $v \in C \setminus F$, $\indf_v(t) = 1$ or (2) for all $v \in C \setminus F$, $\indf_v(t) = 0$. In the former case, we call $C$ \dft{(unanimously) fast}, and in the latter case $C$ is \dft{(unanimously) slow}. For a round $r \in \N$, we say that $C$ is (unanimously) fast (slow) in round $r$ if every $v \in C \setminus F$ is in fast mode (slow mode) for all $t \in [t_r(v), t_{r+1}(v)]$. 

In order to implement the GCS algorithm for cluster clocks, we must show that unanimously fast clusters can ``catch up'' to unanimously slow clusters. This is true for individual nominal clocks: if $v$ is in fast mode and $w$ is in slow mode, then $\hnom_v(t) / \hnom_w(t) \geq (1 + \mu) / (1 + \rho) = \Omega(\mu - \rho)$. So as long as $\mu \gg \rho$, a fast node can always catch up to a slow node.

The story for clusters is, unfortunately, more complicated. Suppose $C$ is (unanimously) fast for an entire round $r$. Even though the individual \emph{nominal} clocks in $C \setminus F$ all run at rates at least $(1 + \varphi)(1 + \mu)$, the amortized rate of $L_C$ may be significantly slower because of the adjustment made to the logical clocks in the Lynch-Welsh step. Specifically, this adjustment could be as large as $e(r) > \mu T(r)$. Thus, the amortized rate of $L_C$ over round $r$ could be as small as $(1 + \varphi)(1 + \mu - e(r) / T(r)) < (1 + \varphi)$. Thus, the logical clock of a cluster in fast mode may increase \emph{slower} than a cluster in slow mode!

To address this potential problem, observe that if a cluster is unanimous, then the nominal clocks for $v \in V$ satisfy $\zeta \leq \hnom_v \leq \zeta \cdot \theta_u$, where $\vartheta_u = 1 + \rho$ and $\zeta = 1+\varphi$ or $(1 + \varphi)(1 + \mu)$ depending on whether the cluster is unanimously slow or fast. Thus, the nominal clock drift between nodes is $O(\rho)$ rather than $O(\mu)$, permitting the cluster synchronization algorithm to converge to a smaller skew. Applying Corollary~\ref{cor:pulse-skew} with $\vartheta = \vartheta_u$ allows us to achieve skews of size $e(r) = O(\rho T(r))$, assuming that the cluster is unanimous for sufficiently many rounds. Therefore, the amortized rate of $L_C$ over round $r$ is at least $(1 + \varphi)(1 + \mu - O(\rho))$. By choosing $\mu$ sufficiently large, but still $O(\rho)$, we can ensure that the amortized rate of $L_C$ is at least $(1 + \varphi)(1 + 7 \mu / 8)$, say. Similarly, we can show that any cluster in slow mode increases at an amortized rate between $(1 + \varphi)(1 - \mu / 8)$ and $(1 + \varphi)(1 + \mu / 8)$, assuming it has been is slow mode for sufficiently long.

To this end, we introduce parameters $e_{f, k}(r)$ and $e_{s, k}$, which give tighter upper bounds on the pulse diameter, assuming that $C$ was unanimously fast or slow  in rounds $r - k, r - k + 1, \ldots, r$. If a node is unanimous for all rounds, we denote these parameters by $e_f(r)$ and $e_s(r)$.

Let $\essg$ denote the steady state error for a general execution of the Lynch-Welsh algorithm---i.e., $\essg = \lim_{r \to \infty} e(r)$. Similarly, we define the steady state error for a unanimous execution (in which nodes are all unanimously fast or slow in all rounds): $\essf = \lim_{r \to \infty} e_f(r)$ and $\esss = \lim_{r \to \infty} e_s(r)$. The following lemma shows that by choosing appropriate values of $\mu$ and round length, we can ensure that $\essg$ is significantly larger than $e_u^\infty$, while still having the best possible asymptotic error in the general case. In order to maintain this gap between (worst case bounds on) general and unanimous steady state error, we must adjust the round length accordingly. We do so by stretching phase 3 (i.e., $\tau_3(r)$) by a factor $c_1$. Specifically, we take
\begin{equation}
  \label{eqn:taus-c}
  \begin{split}
    \tau_1(r) &= \zeta_{\mx} \cdot \vartheta_g \cdot e(r)\\
    \tau_2(r) &= \zeta_{\mx} \cdot \vartheta_g \cdot (e(r) + d)\\
    \tau_3(r) &= c_1 \cdot \zeta_{\mx} \cdot \vartheta_g \cdot (e(r) + U)
  \end{split}
\end{equation}
where $\zeta_\mx = (1 + \varphi)(1 + \mu)$ and $\vartheta_g = (1 + \mu) (1 + \rho)$. The value of $c_1$ will be chosen later, though we remark that its value will be $\Theta(1/\rho)$. Since we will have $c_1 \gg 1$, the assignments~(\ref{eqn:taus-c}) satisfy~(\ref{eqn:taus}) so that we obtain a feasible execution. In particular, we can simply define $\varphi = 1 / c_1$ in order to ensure that~(\ref{eqn:taus}) is also satisfied.

In order to ensure that fast clusters are faster than slow clusters, we must choose $\mu$ to be sufficiently large (as a function of $\rho$). We introduce a parameter $c_2$ and take $\mu = c_2 \cdot \rho$. The following is our main technical result regarding unanimous clusters.

\begin{lem}
  \label{lem:unanimous-error-gap}
  For any $c_2 \geq 32$ and sufficiently small $\rho > 0$, there exist $k = O(1)$ and $c_1 = \Theta(1 / \rho)$ such that for any cluster $C$ and proper execution $X$ with $e(r - k) \leq 2 \essg$ the following hold:
  \begin{enumerate}
  \item If $C$ is unanimously fast for rounds $r - k, r - k +1, \ldots, r$ then for all $v \in C \setminus F$ we have
    \[
    (1 + \varphi)\paren{1 + \frac 7 8 \mu} \leq \frac{L_v(t_v(r + 1)) - L_v(t_v(r))}{t_v(r + 1) - t_v(r)}\,.
    \]
  \item If $C$ is unanimously slow for rounds $r - k , r - k + 1, \ldots, r$, then for all $v \in C \setminus F$ we have
    \[
    (1 + \varphi)\paren{1 - \frac 1 8 \mu} \leq \frac{L_v(t_v(r + 1) - L_v(t_v(r))}{t_v(r + 1) - t_v(r)} \leq (1 + \varphi)\paren{1 + \frac 1 8 \mu}.
    \]
  \end{enumerate}
  If additionally $C$ is unanimously fast (resp.\ slow) in round $r + 1$, then the cluster clock $L_C(t)$ satisfies conclusion 1 (resp.\ 2) above.
\end{lem}

\section{Inter-cluster Algorithm}
\label{sec:intercluster-alg}

In this section we describe an algorithm that synchronizes clocks between adjacent clusters. The algorithm simulates an execution of the gradient clock synchronization (GCS) algorithm~[citations], where each cluster plays the role of a single node in the GCS algorithm. In order to perform the simulation, we define a (logical) cluster clock for each cluster such that the logical clocks of all non-faulty nodes in $C$ are close to the value of $C$'s cluster clock. Nodes in clusters adjacent to $C$ compute estimates of $C$'s cluster clock and use these estimates to adjust their own clock rates according to the GCS algorithm. More precisely, the nodes simulate the GCS algorithm on the cluster graph $\calG$. Each each node $v$ in a cluster $C$ simulates the behavior of $C$ in an execution of GCS using its own logical clock as an estimate of $C$'s cluster clock, as well as its estimates of neighboring cluster clocks.

%###
\paragraph{Fast and slow conditions and triggers.}
%###
Let $\kappa$ be a parameter (to be chosen later), and let $C$ be a cluster. We denote the set of neighboring clusters of $C$ by $N_C$. The following definitions give conditions under which the cluster $C$ should be in fast mode or slow mode in order to implement the GCS algorithm.

\begin{dfn}\label{dfn:fast-condition}
  We say that $C$ satisfies the \dft{fast condition} ($\FC$) at time $t$ if there exists $s \in \N$ such that the following conditions hold:
  \begin{compactitem}
  \item $\FCone$ There exists $A \in N_C$ such that $L_A(t) - L_C(t) \geq 2 s \kappa$.
  \item $\FCtwo$ For all $B \in N_C$, $L_C(t) - L_B(t) \leq 2 s \kappa$.
  \end{compactitem}
\end{dfn}

\begin{dfn}\label{dfn:slow-condition}
  We say that $C$ satisfies the \dft{slow condition} ($\SC$) at time $t$ if there exists $s \in \N$ such that the following two conditions hold:
  \begin{compactitem}
  \item $\SCone$ There exists $A \in N_C$ such that $L_C(t) - L_A(t) \geq (2 s - 1) \kappa$.
  \item $\SCtwo$ For all $B \in N_C$, $L_B(t) - L_C(t) \leq (2 s - 1) \kappa$.
  \end{compactitem}
\end{dfn}

%% The proof of the following lemma is straightforward (cf. [Kuhn and Oshman]):

%% \begin{lem}\label{lem:fast-slow-mutex}
%%   The conditions $\FC$ and $\SC$ are mutually exclusive. That is, if $C$ satisfies $\FC$, then $C$ does not satisfy $\SC$, and vice versa.
%% \end{lem}

In order for the (analysis of the) GCS algorithm to work, we must guarantee that if a cluster satisfies the fast (resp.\ slow) condition, then the cluster is unanimously fast (resp.\ slow). Further, in order to maintain the guarantees on \emph{amortized} rates of cluster clocks shown in Section~\ref{sec:cluster-alg}, our implementation of the fast and slow conditions should guarantee that clusters remain in fast or slow mode for the entire duration of any round in which the respective condition is satisfied. In order to implement the fast and slow conditions subject to these constraints, as well uncertainty in nodes' estimates of neighboring cluster clocks, we define the following triggers. The parameter $\delta$ will be chosen later.

\begin{dfn}\label{dfn:fast-trigger}
  We say that $v \in C$ satisfies the \dft{fast trigger} ($\FT$) at time $t$ if the following two conditions hold:
  \begin{compactitem}
  \item $\FTone$ There exists $A \in N_C$ such that $\tilL_A^v(t) - L_v(t) \geq 2 s \kappa - \delta$.
  \item $\FTtwo$ For all $B \in N_C$, $L_v(t) - \tilL_B^v(t) \leq 2 s \kappa + \delta$.
  \end{compactitem}
\end{dfn}

\begin{dfn}\label{dfn:slow-trigger}
  We say that $v \in C$ satisfies the \dft{slow trigger} ($\ST$) at time $t$ if the following two conditions hold:
  \begin{compactitem}
  \item $\STone$ There exists $A \in N_C$ such that $L_v(t) - \tilL_A^v(t) \geq (2 s - 1) \kappa - \delta$.
  \item $\STtwo$ For all $B \in N_C$, $\tilL_B^v(t) - L_v(t) \leq (2 s - 1) \kappa + \delta$.
  \end{compactitem}
\end{dfn}

The following lemma shows that for all $\delta < 2 \kappa$ the triggers above cannot both be simultaneously satisfied. In particular, taking $\delta = 0$, the following lemma implies that the fast and slow \emph{conditions} are also mutually exclusive.

\begin{lem}\label{lem:fast-slow-mutex}
  The conditions $\FT$ and $\ST$ are mutually exclusive. That is, if $C$ satisfies $\FT$, then $C$ does not satisfy $\ST$, and vice versa.
\end{lem}
With the fast and slow triggers defined, we can describe the inter-cluster algorithm.

\begin{wrapfigure}[10]{R}{0.5\textwidth}
  \begin{algorithm}[H]
    \label{alg:intercluster}
  \caption{$\InterclusterSync(v, \kappa, \delta, T)$}
  \ForEach{round $r \in \N$}{
    \AtTime{$L_v(t_v(r))$}{
      \uIf{$v$ satisfies $\FT$}{
        $\indf_v \leftarrow 1$\;
      }
      \uElseIf{$v$ satisfies $\ST$}{
        $\indf_v \leftarrow 0$\;
      }
    }
  }
\end{algorithm}  
\end{wrapfigure}
$\InterclusterSync$ differs from other descriptions of the GCS algorithm~\cite{kuhn18gradient,lenzen10tight} in that it can only switch from fast to slow mode, or vice versa, at predetermined discrete times. For the analysis of the GCS algorithm, we require that the fast (resp.\ slow) trigger implements the fast (resp.\ slow) condition in the sense that whenever the condition is satisfied, the corresponding trigger is also satisfied. For \emph{clusters} we require more: Even if a cluster is unanimously fast or slow, we cannot immediately infer sufficiently tight bounds on the rate of $L_C$. Instead, we must apply Lemma~\ref{lem:unanimous-error-gap}. In particular, we must wait until $k = O(1)$ unanimous rounds have elapsed until we can guarantee sufficiently tight bounds on the rate of $L_C$ to apply the GCS algorithm analysis as a black box. The following definition describes sufficient conditions under which $\InterclusterSync$ guarantees that a node has been in in fast (resp.\ slow) mode ``sufficiently long'' whenever the fast (resp.\ slow) condition is satisfied.

\begin{dfn}\label{dfn:faithful}
  Let $T : \N \to \R$ be a sequence of round lengths, $C$ a cluster, and $v \in C$. For every $r \in \N$, let $t_v(r)$ denote the time at which $v$ begins round $r$ in an execution of the $\ClusterSync$ algorithm. Let $k$ be a constant such that the conclusion of Lemma~\ref{lem:unanimous-error-gap} is satisfied. For any time $t$, let $r_t = \max\set{r \sucht t_v(r) \leq t}$. We say that an execution of $\InterclusterSync$ is \dft{faithful for $v$} if the following conditions hold:
  \begin{compactitem}
  \item For all $t \in \R^+$ such that $C$ satisfies $\FC$ at time $t$, $v$ satisfies $\FT$ at all $t' \in [t_v(r_t - k), t_v(r_t)]$.
  \item For all $t \in \R^+$ such that $C$ satisfies $\SC$ at time $t$, $v$ satisfies $\ST$ at all $t' = [t_v(r_t - k), t_v(r_t)]$.
  \end{compactitem}
  We say that the execution is \dft{faithful for $C$} if it is faithful for every node $v \in C \setminus F$.
\end{dfn}

Applying Definition~\ref{dfn:faithful} we obtain the following consequence of Lemma~\ref{lem:unanimous-error-gap}:

\begin{cor}
  \label{cor:faithful-error-gap}
  Suppose $X$ is a faithful execution for $C$. Then for every $t \in \R$ the following holds. If $C$ satisfies $\FC$ at time $t$ then
  \[
  (1 + \varphi)\paren{1 + \frac 7 8 \mu} \leq \frac{L_v(t_v(r_t + 1)) - L_v(t_v(r_t))}{t_v(r_t + 1) - t_v(r_t)}.
  \]
  If $C$ satisfies $\SC$ at time $t$ then
  \[
  (1 + \varphi)\paren{1 - \frac 1 8 \mu} \leq \frac{L_v(t_v(r_t + 1) - L_v(t_v(r_t))}{t_v(r_t + 1) - t_v(r_t)} \leq (1 + \varphi)\paren{1 + \frac 1 8 \mu}
  \]
\end{cor}

In the next paragraph, we will show that Corollary~\ref{cor:faithful-error-gap} is strong enough that we can apply the analysis of the GCS algorithm as a black box to bound skew between adjacent clusters. In the remainder of this section, we will give sufficient conditions---in particular, choices of the parameters $\kappa$ and $\delta$---under which every execution is guaranteed to be faithful. We fix $k$ to be a constant such that the conclusion of Lemma~\ref{lem:unanimous-error-gap} is satisfied.

\begin{lem}
  \label{lem:faithful-delta}
  Suppose $T$ and $\Err$ satisfy the hypotheses of Proposition~\ref{prop:cluster-error-bound}, and that $k$ is sufficiently large that the conclusion of Lemma~\ref{lem:unanimous-error-gap} holds. Let $\delta = (k + 5) \Err$ and $\kappa = 3 \delta$. Then for every cluster $C \in \calC$, every execution $X$ is faithful for $C$.
\end{lem}

%###
\paragraph{$\InterclusterSync$ Simulates GCS.}
%###
We now show that a faithful execution $X$ of $\InterclusterSync$ on the (physical) network $G$ \emph{simulates} an execution $\bar{X}$ of the GCS algorithm on the network $\calG$ in a sense made precise below. As a result, we can apply the analysis of the GCS algorithm on $\bar{X}$ to derive bounds on the local skew between adjacent cluster clocks in $X$. Before defining simulation formally, we recall the axioms required by the GCS algorithm.

\begin{dfn}
  \label{gcs:axioms}
  Suppose $\calG = (\calC, \calE)$ is a network, and each $C \in \calC$ computes a logical clock $L_C : \R \to \R$. We say that $\set{L_C}$ satisfy the \dft{GCS axioms} if there exist constants $\rho, \mu > 0$ such that the following hold for all times $t \in \R$ and nodes $C \in \calC$:
  \begin{compactitem}
  \item[(A1)] $1 \leq \frac{\mathrm{d}}{\mathrm{dt}} L_v(t) \leq (1 + \rho)(1 + \mu)$.
  \item[(A2)] If $C$ satisfies $\SC$ at time $t$, then $\frac{\mathrm{d}}{\mathrm{dt}} L_v(t) \leq 1 + \rho$.
  \item[(A3)] If $C$ satisfies $\FC$ at time $t$ then $1 + \mu \leq \frac{\mathrm{d}}{\mathrm{dt}} L_v(t)$
  \item[(A4)] $\mu / \rho > 1$.
  \end{compactitem}
  Any execution of any algorithm satisfying these axioms is said to \dft{implement GCS}.
\end{dfn}

\begin{thm}[\cite{kuhn18gradient}]
  \label{thm:gcs}
  Suppose an algorithm $A$ implements GCS and the global skew is bounded by $\mathcal{S}$. Then for all $(B, C) \in \calE$ and sufficiently large $t$ we have
  $\abs{L_B(t) - L_C(t)} = O(\kappa\log_{\mu / \rho}\mathcal{S})$.
\end{thm}

In general, an execution of our algorithm does not satisfy the GCS axioms for the values of $\rho$ and $\mu$ as specified in the previous sections. However, for suitable choices of these parameters, we can find different parameters $\bar{\rho}$, $\bar{\mu}$ for which our logical clocks do satisfy the GCS axioms. In the appendix, we prove the following.

\begin{prop}
  \label{prop:is-simulates-gcs}
  Suppose $X$ is a faithful execution of $\InterclusterSync$ on $G$. Then the cluster clocks $\set{L_C \sucht C \in \calC}$ satisfy the GCS axioms for $\bar{\rho} = (1 + \varphi)(1 + (1/4)\mu) - 1$ and $\bar{\mu} = (1 + \varphi)(1 + (7/8) \mu) - 1$.
\end{prop}

%% \begin{cor}
%%   \label{cor:is-simulates-gcs}
%%   Suppose the hypotheses of Proposition~\ref{prop:is-simulates-gcs} are satisfied. Then $\InterclusterSync$ guarantees that for all adjacent clusters $(B, C) \in \calE$ and all times $t$,
%%   \[
%%   \abs{L_B(t) - L_C(t)} \leq \text{something}
%%   \]
%% \end{cor}

%###
\paragraph{Proof of Theorem~\ref{thm:main}}
%###
So far, we have neglected the global skew. Bounding it provides, essentially, the induction base for proving a small local skew. As no new techniques are required for ensuring a global skew of $O(\delta D)$, where $D$ is the hop diameter of both $G$ and $\cal{G}$, we sketch a feasible construction in Appendix~\ref{app:gcs}.

We finally have all the pieces in place to prove our main result, Theorem~\ref{thm:main}. We assume the parameters $\rho$, $d$, and $U$ are given, and that $\rho$ is sufficiently small so that the conclusions of all required lemmas hold. Specifically, we take:
\begin{equation}
  \begin{split}
      \mu &= c_2 \cdot \rho, \qquad \tau_1 = \thg \cdot \Err, \qquad \tau_2 = \thg \cdot (\Err + d), \qquad \tau_3 = \thg \cdot c_1 \cdot (\Err + U),\\
  c_1 &= \frac 1 \varphi = \frac{(1/2) - \e}{1 + c_2} \cdot \frac 1 \rho, \qquad
  c_2 = 32, \qquad
  \e = 1 / 4096.
  \end{split}
\end{equation}

\begin{proof}[Proof of Theorem~\ref{thm:main}]
  Let $v, w \in V \setminus F$ with $\set{v, w} \in E$. We first consider the case where $v, w \in C$. We have $\Err = \beta_g / (1 - \alpha_g)$ where these values are computed in Claim~\ref{claim:gen-error-bound}. In particular, applying the conclusion of the claim, we get $\Err = O(\rho \cdot d + U)$. Then by Corollary~\ref{cor:cluster-skew-bound}, we get
  $\abs{L_v(t) - L_w(t)} \leq 2\thg \cdot \Err = O(\rho \cdot d + U)$.

  Now consider the case where $v \in B$ and $w \in C$ where $(B, C) \in \calE$. By Theorem~\ref{thm:gcs}, Proposition~\ref{prop:is-simulates-gcs}, and Theorem~\ref{thm:global_skew} (which states that the global skew is $O(\delta D)$), we obtain
  \[
  \abs{L_B(t) - L_C(t)} = O(\log(\delta D) \kappa) = O((\rho \cdot d + U) \log D).
  \]
  We then bound
  \begin{align*}
    \abs{L_v(t) - L_w(t)} &\leq \abs{L_v(t) - L_B(t)} + \abs{L_B(t) - L_C(t)} + \abs{L_c(t) - L_w(t)}\\
    &= O((\rho \cdot d + U) \log D) + 2 \thg \cdot \Err\\
    &= O((\rho \cdot d + U) \log D),
  \end{align*}
  which gives the desired conclusion.
\end{proof}

\bibliographystyle{plain}
\bibliography{fault-tolerant-gcs}

\appendix

\section{Further Related Work}\label{app:further}

A basic synchronization algorithm that can cope with Byzantine faults is the one by Srikanth and Toueg~\cite{srikanth87optimal}. In a fully connected network, it maintains synchronization among the nodes by a propose-and-pull mechanism: based on a timeout, nodes will propose to resynchronize by generating a corresponding local event upon having received at least $n-f$ respective messages; however, $f+1\leq n-2f$ propose messages are sufficient to cause ``late proposers'' to send propose messages even if their respective timeouts are not expired yet. This achieves asymptotically optimal~\cite{lundelius84upper} skew of $O(d)$ despite $f<3n$ Byzantine faults\dash---provided that there is no guaranteed \emph{lower} bound on the communication delay. The Lynch-Welch algorithm~\cite{welch88new} improves the skew to $O(u+(\vartheta-1)d)$ under the additional assumption that messages are underway for at least $d-u$ time; again, this is asymptotically optimal under these assumptions.\footnote{This again follows from~\cite{lundelius84upper}, together with a simple indistinguishability argument for the $(\vartheta-1)d$ term.} The algorithm achieves this by simulating synchronous rounds, each of which is used to perform an approximate agreement~\cite{dolev86reaching} step on when the round should have started and adjusting clocks accordingly.

Both algorithms could be employed in our construction, where we chose Lynch-Welch for its better skew. Both algorithms also share the characteristic that, in their basic variants, logical clocks ``jump'' to implement phase corrections, which is incompatible with the requirement that logical clocks satisfy lower and upper bounds on their rates in a GCS algorithm. This issue is easily addressed by amortizing clock adjustments over sufficient periods of time~\cite{lamport85synchronizing}. This requires to adjust the ``round length'' of these algorithms, but this change has no asymptotic impact on skews\dash---neither for the plain algorithms nor in our construction.

A series of works considers synchronization algorithms that are simultaneously resilient to $f<n/3$ Byzantine faults and self-stabilizing, i.e., synchronization is re-established despite the (ongoing) interference from Byzantine faulty nodes after transient faults cease. Dolev and Welch~\cite{dolev04stabilizing} proposed the problem, proving that it can actually be solved. However, their algorithm has exponential stabilization time, i.e., $2^{\Omega(f)}d$ time may pass after transient faults cease before the logical clocks meet the synchronization and progress requirements (again). The stabilization time was improved to polynomial~\cite{daliot03self-stabilizing}, then linear~\cite{dolev07bounded}, and finally (randomized) logarithmic~\cite{lenzen17easy}. The latter construction transforms any synchronous $R$-round consensus algorithm into a solution to the problem that stabilizes in $O(R\log n)$ time and sends $O(M\log n)$ bits over each link in $\Theta(d)$ time, where $M$ is the message size of the consensus algorithm. If the consensus algorithm is randomized, the transformation works the same way, but the stabilization time bound holds with high probability (instead of deterministically). Beside the smallest known stabilization time, this transformation also yields the best known trade-offs between stabilization time and amount of communication. All of these algorithms have in common that they achieve $O(d)$ skew, as they rely on the propose-and-pull mechanic underlying the Srikanth-Toueg algorithm. However, such algorithms can be used to make the Lynch-Welch algorithm self-stabilizing, by utilizing the inaccurate (and typically also infrequent) synchronization events to ``jump-start'' the simulation of synchronous approximate agreement rounds the Lynch-Welch algorithm is based on~\cite{khanchandani18self}. The result is a routine that combines the extreme resilience of the self-stabilizing routine with the asymptotically optimal skew of the Lynch-Welch algorithm.

To date, GCS has been studied in fault-free networks only. The problem was introduced by Fan and Lynch~\cite{lynch2004gradient}, alongside a surprising lower bound of $\Omega(\log D/\log \log D)$ on the local skew. The first non-trivial upper bound of $O(\sqrt{D})$ on the local skew that can be achieved is due to Locher and Wattenhofer~\cite{locher2006oblivious}. This algorithm has nodes try catching up with the maximum logical clock value among their neighbors, but under the constraint that they never run faster than their hardware clock rate when there is a neighbor whose clock lags $\Theta(\sqrt{D})$ or more behind. As the global skew is bounded by $O(D)$, at most $O(\sqrt{D})$ consecutive nodes can be ``blocked'' from catching up, implying that an individual node is not prevented from doing so for more than $\Theta(\sqrt{D})$ time; this gives rise to the bound on the skew. Subsequently, the tight bound of $\Theta(\log D)$ on the local skew mentioned earlier has been established~\cite{lenzen10tight}. The respective algorithm can be seen as switching between the ``catching up'' and ``blocking'' strategy more than once, by comparing for some suitably chosen $\kappa$ the largest $s\in \mathbb{N}_0$ such that some neighbor's clock is at least $s\kappa$ ahead to the largest $s'\in \mathbb{N}_0$ such that some neighbor's clock is at least $s'\kappa$ behind. One can then show that the length of paths with sufficient skew to ``block'' nodes from catching up decreases exponentially with $s$, yielding the stated bound.

Astonishingly, the algorithmic approach turns out to be quite robust and flexible. The algorithm generalizes to networks in which edges $e=\{v,w\}$ have weight $\epsilon_e$ indicating the accuracy with which $v$ and $w$ can estimate each other's clock values, by doing nothing more than choosing $\kappa$ proportional to $\epsilon_e$. Moreover, the algorithm is almost self-stabilizing, in the sense that it will re-establish its local skew bound from any state in $O(\mathcal{S}/\mu)$ time, provided that a global skew bound of $\mathcal{S}$ is satisfied. As logical clocks must not run more than factor $1+\mu$ faster than hardware clocks, this time bound is optimal so long as we do not allow violating this bound on the rate. Moreover, this stabilization property can be leveraged to allow for dynamic topologies. That is, edges may appear and disappear in a worst-case fashion, yet the algorithm must maintain its skew bounds on all paths that consist only of edges that have been present for $\Omega(\mathcal{S}/\mu)$ time. Adding a mechanism to carefully ``activate'' the consideration of newly arriving edges level by level (i.e., for increasing values of $s$) in a well-timed fashion, the algorithm guarantees this property with no further modification~\cite{kuhn10gradient,kuhn18gradient}. In addition, choosing $\mu \in \Theta(1)$ and using that the algorithm achieves $\mathcal{S}\in O(D)$, where $D$ is the (weighted, dynamic) diameter of the graph, we see that the algorithm stabilizes new edges in $O(D)$ time. Again, this bound is worst-case optimal~\cite{kuhn18gradient}. Note that the dynamic version of the algorithm in particular shows that crash failures can be tolerated, as repeatedly checking liveness of nodes (which is implicit, as estimating clock values necessitates communication) enables mapping of crash failures to deleting all incident links of the crashed node. In this work, we provide the first (non-trivial) GCS algorithm resilient to non-benign faults. As it is based on the same algorithmic concept and a generic construction, we anticipate that all of the results just mentioned can be carried over, even though we confine ourselves to the static setting in this paper.

\section{Analysis of the Cluster Synchronization Algorithm}\label{app:lw}

\begin{proof}[Proof of Lemma~\ref{lem:nominal-round-length}]
Consider any times $t<t'$ such that, for some fixed $\delta$, it holds that $\delta_v(\tau)=\delta$ for all $\tau \in [t,t')$. By definition of the logical clock rate, we have that
  \begin{align*}
    L_v(t') - L_v(t)
    &= \int_{t}^{t'} (1+\varphi\cdot\delta) (1+\mu\cdot\indf_v(\tau)) h_v(\tau)\, d\tau\\
    &= \frac{1+\varphi\cdot \delta}{1+\varphi}\cdot \int_{t}^{t'} (1+\varphi)(1+\mu\cdot\indf_v(\tau)) h_v(\tau)\, d\tau\\
    &= \frac{1+\varphi\cdot \delta}{1+\varphi}\cdot\int_{t}^{t'} \hnom_v(\tau)\,d\tau\,.
  \end{align*}
Denote by $s_v(r)$ the Newtonian time when phase 2 of round $r$ ends at node $v$. Algorithm~\ref{alg:cluster-sync} stipulates that $L_v(t_v(r+1))-L_v(s_v(r))=\tau_3$ and $L_v(s_v(r))-L_v(t_v(r))=\tau_1(r)+\tau_2(r)$. Moreover,
  \begin{equation*}
  \delta_v(t) = \begin{cases}
  1 & \mbox{if } t_v(r)\leq t\leq s_v(r)\\
  1 - \paren{1 + \frac{1}{\varphi}} \frac{\Delta_v(r)}{\tau_3(r)+\Delta_v(r)} & \mbox{if }s_v(r)\leq t \leq t_v(r+1)\,.
  \end{cases}
  \end{equation*}
Hence, from the above calculation we get that
\begin{equation*}
\tau_1(r)+\tau_2(r) = L_v(s_v(r))-L_v(t_v(r)) = \int_{t_v(r)}^{s_v(r)} \hnom_v(\tau)\,d\tau
\end{equation*}
and
\begin{align*}
\tau_3(r) &= L_v(t_v(r+1))-L_v(s_v(r))\\
&= \frac{1+\varphi\cdot\left(1 - \paren{1 + \frac{1}{\varphi}} \frac{\Delta_v(r)}{\tau_3(r)+\Delta_v(r)}\right)}{1+\varphi}
\cdot\int_{s_v(r)}^{t_v(r+1)} \hnom_v(\tau)\,d\tau\\
&= \frac{\tau_3(r)}{\tau_3(r)+\Delta_v(r)}\cdot\int_{s_v(r)}^{t_v(r+1)} \hnom_v(\tau)\,d\tau\,,
\end{align*}
i.e., $\tau_3(r)+\Delta_v(r) = \int_{s_v(r)}^{t_v(r+1)} \hnom_v(\tau)\,d\tau$. We conclude that
  \begin{equation*}
    \int_{t_v(r)}^{t_v(r + 1)} \hnom_v(\tau)\,d\tau = \tau_1(r)+\tau_2(r)+\tau_3(r)+\Delta_v(r) = T(r) + \Delta_v(r)\,.\qedhere
  \end{equation*}
\end{proof}

\subsection{Slow-down Simulation}

In several recent works on clock synchronization, it is assumed that the rates hardware clocks satisfy $1 \leq h(t) \leq \vartheta$ for some constant $\vartheta > 1$. In analyzing our algorithm, it will be helpful to consider more general bounds on clock rates. In particular, we will consider ``hardware'' clocks whose rates are artificially increased by a \emph{speedup} factor $\zeta > 1$. In order to use the results of~\cite{kuhn18gradient,lenzen10tight,welch88new} without modification, we prove the following lemma, which shows that executions in which hardware clocks are sped up by a factor $\zeta$ are indistinguishable from executions with hardware clock drift in the range $[1, \vartheta]$.

\begin{lem}
  \label{lem:speedup}
  Let $A$ be any algorithm, and let $d, U \in \R$ denote the maximum message delay and delay uncertainty, respectively. Suppose that for all $v \in V$, the rate of $v$'s hardware clock satisfies $\zeta \leq h_v(t) \leq \zeta \vartheta$ for some $\zeta > 0$ and $\vartheta > 1$. Then for every execution $X$ of $A$, there is an indistinguishable execution $\overline{X}$ of $A$ such that the following hold:
  \begin{enumerate}
  \item For all $v \in V$ and $t \in \R$, $1 \leq h_v(t) \leq \vartheta$.
  \item The maximum message delay is $\overline{d} = \zeta d$.
  \item The maximum message delay uncertainty is $\overline{U} = \zeta U$.
  \end{enumerate}
\end{lem}
\begin{proof}
  Consider the transformation $f : \R \to \R$ given by $f(t) = \zeta t$. For every event $\eta$ occurring in $X$ at Newtonian time $t$, the corresponding event $\overline{\eta}$ occurs at time $f(t) = \zeta t$ in $\overline{X}$. The hardware clocks in $\overline{X}$ are defined by $\overline{h}_v(t) = \frac{1}{\zeta} h_v(\frac{t}{\zeta})$. The resulting execution $\bar{X}$ satisfies the above bounds on hardware clock rates and delays by construction. To see that the executions $X$ and $\overline{X}$ are indistinguishable, we must show that the hardware time of all events are the same for all processors. By a change of variables, we get for each time $t$ that
  \[
  \int_{0}^{f(t)} \overline{h}_v(\tau)\, d\tau = \int_0^{\zeta t} \overline{h}_v(\tau)\, d\tau = \int_0^{\zeta t} \frac{1}{\zeta} h_v\!\left(\frac{\tau}{\zeta}\right)\, d\tau = \int_0^t h_v(\tau)\, d\tau\,,
  \]
showing that this holds true as well.
\end{proof}

\begin{dfn}
  \label{dfn:sigma-reduced}
  For an execution $X$ in which $\zeta \leq h_v(t) \leq \zeta \vartheta$ for all $v \in V$, $t \in \R$, we call the execution $\bar{X}$ postulated by Lemma~\ref{lem:speedup} the \dft{$\zeta$-reduced} execution of $X$.
\end{dfn}

\subsection{Proper executions.}

In order for Algorithm~\ref{alg:cluster-sync} to function as intended, the broadcast pulses of all correct processors in a cluster must be sent and received in the prescribed phases of each round. Additionally, we will require that logical clocks maintain a minimal rate of at least $1$. The following definition gives sufficient conditions ensuring that a round is executed properly.
\begin{dfn}
  \label{dfn:cluster-sync-proper-ex}
  Fix a cluster $C$, a round $r$. We say that round $r$ of $\ClusterSync$ (Algorithm~\ref{alg:cluster-sync}) is \dft{properly executed} if the following conditions hold for all $v, w \in C \setminus F$:
  \begin{compactenum}
  \item $v$ is in round $r$ when $w$ broadcasts its round $r$ pulse
  \item $v$ receives $w$'s round $r$ pulse before logical time $L_v(t_v(r)) + \tau_1(r) + \tau_2(r)$ (line~\ref{cs:phase-2})
  \item $\abs{\Delta_v(r)} \leq \varphi \cdot \tau_3(r)$ (line~\ref{cs:set-Delta})
  \end{compactenum}
\end{dfn}

Here, we derive bounds on the skew between logical clocks $L_v$ within a cluster assuming that all rounds are properly executed. In the sequel, we will give sufficient conditions for a sequence of rounds to be properly executed. We begin by bounding the rate of logical clocks $L_v(t)$.

\begin{lem}
  \label{lem:logical-clock-rate}
  Suppose that round $r$ is properly executed for a cluster $C$. Then for all $t \in \R$ satisfying $t_v(r) \leq t \leq t_v(r + 1)$ we have
  $0 \leq \delta_v(t) \leq \frac{2}{1 - \varphi}$.
  Moreover, for all $t,t' \in \R$ satisfying $t_v(r) \leq t \leq t' \leq t_v(r + 1)$, we have
  \[
  t' - t \leq L_v(t') - L_v(t) \leq \paren{1 + \frac{2 \varphi}{1 - \varphi}} (1 + \mu) (1 + \rho) (t' - t)\,.
  \]
\end{lem}
\begin{proof}
Assume that $s_v(t)$ is the Newtonian time at which $v$ ends phase $2$ of round $r$. Thus, for $t_v(r)\leq t\leq s_v(r)$ it holds that $\delta_v(t) = 1$ and $0 \leq \delta_v(t) \leq \frac{2}{1 - \varphi}$ follows from $\abs{\varphi}<1$. Now consider a time $t$ with $s_v(r)\leq t\leq t_v(r+1)$. For such a time $t$, we have that $\delta_v(t)=1 - \paren{1 + \frac{1}{\varphi}} \frac{\Delta_v(r)}{\tau_3(r) + \Delta_v(r)}$. As round $r$ is properly executed, $\abs{\Delta_v(r)}\leq \varphi \cdot \tau_3(r)$. Using this bound, we obtain
  \begin{equation*}
    \delta_v(t) \geq 1 - \paren{1 + \frac{1}{\varphi}} \frac{\varphi\tau_3(r)}{\tau_3(r) + \varphi\tau_3(r)}
    = 1 - \paren{1 + \frac{1}{\varphi}} \frac{\varphi}{1 + \varphi}
    = 1 - \frac{1 + \varphi}{1 + \varphi} = 0
  \end{equation*}
  and
  \begin{equation*}
    \delta_v(t) \leq 1 - \paren{1 + \frac{1}{\varphi}} \frac{-\varphi\tau_3(r)}{\tau_3(r) - \varphi\tau_3(r)}
    =1+\frac{1+\varphi}{1-\varphi}
    =\frac{2}{1 - \varphi}\,,
  \end{equation*}
establishing the stated bound on $\delta_v(t)$ for $t_v(r)\leq t \leq t_v(r+1)$.

Recall that \eqref{eqn:logical-clock} implies that $L_v(t')-L_v(t)= \int_t^{t'}(1+\varphi \cdot \delta_v(\tau))(1+\mu \cdot \indf_v(\tau))h_v(\tau)\,d\tau$. Moreover, $\indf_v(t)\in \{0,1\}$ and $h_v(t)\in [1,1+\rho]$ for all times $\tau$ by definition. Inserting these bounds into the integral, for $t_v(r) \leq t \leq t' \leq t_v(r + 1)$ we obtain
\begin{align*}
  t'-t =\int_t^{t'}1\,d\tau 
  &\leq \int_t^{t'}(1+\varphi \cdot \delta_v(\tau))(1+\mu \cdot \indf_v(\tau))h_v(\tau)\,d\tau\\
  &\leq \int_t^{t'}\paren{1 + \frac{2 \varphi}{1 - \varphi}} (1 + \mu) (1 + \rho)\,d\tau\\
  &=\paren{1 + \frac{2 \varphi}{1 - \varphi}} (1 + \mu) (1 + \rho) (t' - t)\,,
\end{align*}
i.e., the second claim of the lemma holds.
\end{proof}

\begin{ntn}
  In accordance with Lemma~\ref{lem:logical-clock-rate}, we denote
  \begin{equation}
    \label{eqn:gamma}
    \thmax = \paren{1 + \frac{2 \varphi}{1 - \varphi}} (1 + \mu) (1 + \rho).
  \end{equation}
\end{ntn}

\begin{lem}
  \label{lem:logical-clock-round}
  Suppose that all rounds $r \in \N$ are properly executed for $C\in \mathcal{C}$. For each $r$, define $\mathcal{T}(r) = \sum_{i=1}^{r-1} T(i)$. Then for all $v \in C$ and $r \in \N$ we have
  \[
  L_{v}(t_v(r)) =  \mathcal{T}(r) \quad\text{and}\quad L_{v}(p_v(r)) = \mathcal{T}(r) + \tau_1(r)\,.
  \]
\end{lem}
\begin{proof}
We proof the claim by an induction on $r$. For $r=0$, $L_v(t_v(1))=0=\mathcal{T}(1)$ by initialization, cf.~Algorithm~\ref{alg:cluster-sync}. For the induction step consider $r>1$.  $L_v(t_v(r-1))=\mathcal{T}(r-1)$. The algorithm starts round $r$ at logical time $L_v(t_v(r-1))+\tau_1+\tau_2+\tau_3=L_v(t_v(r-1))+T(r-1)$. Because round $r-1$ is properly executed and $|\varphi|<1$, Lemma~\ref{lem:logical-clock-rate} guarantees that the logical clock of $v$ (continuously) increases at rate at least $1$ during round $r-1$ at $v$. Hence, it reaches this value and $t_v(r)$ is well-defined. By the induction hypothesis applied to round $r-1$, this yields 
  \begin{align*}
    L_v(t_v(r))=L_v(t_v(r-1))+T(r-1)=\mathcal{T}(r-1)+T(r-1)=\mathcal{T}(r)\,.
  \end{align*}
  The second claim follows directly from the first:
  \begin{align*}
    L_v(p_v(r)) &= L_v(t_v(r))+\tau_1(r)=\mathcal{T}(r)+\tau_1(r)\,. \qedhere
  \end{align*}
\end{proof}

\begin{dfn}
  Fix a cluster $C$ and a round $r \in \N$. We denote the multi-set of Newtonian times at which correct nodes in $C$ produce their round $r$ pulse by $\bfp_C$. That is
  \[
  \bfp_C(r) = \set{p_v(r) \sucht v \in C \setminus F}\,.
  \]
  We define the round $r$ \dft{pulse diameter} of $C$, denoted $\norm{\bfp_C(r)}$, by
  \[
  \norm{\bfp_C(r)} = \max \bfp_C(r) - \min \bfp_C(r)\,.
  \]
  If the cluster $C$ is clear from context, we will omit the subscript $C$ from the notation above.
\end{dfn}

\begin{lem}
  \label{lem:cluster-clock-skew}
  Suppose that all rounds $r \in \N$ are properly executed for $C\in \mathcal{C}$. Then for all $v, w \in C \setminus F$ and $t \in \R^+$ we have
  \[
  \abs{L_v(t) - L_w(t)} \leq \thmax \cdot \norm{\bfp(r)} + (\thmax - 1) \paren{T(r) + \tau_1(r + 1) - \tau_1(r)},
  \]
  where $r$ is the largest round such that $\max \bfp(r) \leq t$ (with the convention $\bfp(0)=\{0,\ldots,0\}$).
\end{lem}
\begin{proof}
By definition of $r$, $t \leq \max \bfp(r+1)$. By Lemma~\ref{lem:logical-clock-rate}, logical clocks progress at rates between $1$ and $\thmax$. In particular, $\max \bfp(1)\leq \tau_1$ and $|L_v(t)-L_w(t)|\leq (\thmax-1)\cdot t + |L_v(0)-L_w(0)|=(\thmax-1)\cdot t$, because logical clocks are initialized to $0$. This shows the claim for the special case of $r=0$.

For the case $r\in \N$, in addition to the bounds on logical clock rates we use that $L_v(p_v(r))=L_w(p_w(r))$ by Lemma~\ref{lem:logical-clock-round}. Assuming w.l.o.g.\ that $L_v(t)\geq L_w(t)$, this yields that
\begin{align*}
|L_v(t)-L_w(t)|&=L_v(t)-L_v(p_v(r))+L_w(t)-L_w(p_v(r))+L_v(p_v(r))-L_w(p_v(r))\\
&\leq (\thmax-1)(t-p_v(r))+L_w(p_w(r))-L_w(p_v(r))\\
&\leq (\thmax-1)(t-p_v(r))+p_w(r)-p_v(r)\\
&= (\thmax-1)\cdot t-\thmax \cdot p_v(r)+p_w(r)\\
&\leq (\thmax-1)\cdot \max \bfp(r+1)-\thmax \cdot \min \bfp(r)+\max \bfp(r)\\
&= (\thmax-1)(\max \bfp(r+1)-\min \bfp(r))+\thmax(\max \bfp(r)-\min \bfp(r))\\
&\leq (\thmax-1)(\max \bfp(r+1)-\max \bfp(r))+\thmax \norm{\bfp(r)}\,.
\end{align*}
Therefore, the claim follows if $\max \bfp(r+1)-\max \bfp(r)\leq T(r)+\tau_1(r+1)-\tau_1(r)$. To see this, let $u\in C$ be the node such that $p_u(r+1)=\max \bfp(r+1)$. By Algorithm~\ref{alg:cluster-sync}, the logical duration $L_u(p_u(r+1))-L_u(p_u(r))=\tau_2(r)+\tau_3(r)+\tau_1(r+1)=T(r)+\tau_1(r+1)-\tau_1(r)$. Using again that logical clock rates are at least $1$ by Lemma~\ref{lem:logical-clock-rate}, we infer that
\begin{align*}
\max \bfp(r+1)-\max \bfp(r) &=p_u(r+1)-p_u(r)+p_u(r)-\max \bfp(r)\\
&\leq L_u(p_v(r+1))-L_u(p_v(r))\\
&=T(r)+\tau_1(r+1)-\tau_1(r)\,.\qedhere
\end{align*}
\end{proof}

Finally, we state a result (due to Khanchandani and Lenzen~\cite{khanchandani18self}) that gives an upper bound on the difference between the (Newtonian) times at which clocks within the same cluster pulse.

\begin{cor}[of Lenzen and Khanchandani~\cite{khanchandani18self}, Corollary~4]
  \label{cor:pulse-skew-orig}
  Fix a cluster $C$ and a round $r$, and suppose round $r$ of Algorithm~\ref{alg:cluster-sync} is properly executed. Suppose further that for each $v \in C \setminus F$, $v$'s nominal clock rate satisfies $1 \leq \hnom_v(t) \leq \vartheta$ for all $t \in [\min \bfp_C(r), p_v(r + 1)]$. Then
  \begin{equation*}
    \norm{\bfp_C(r + 1)} \leq \frac{2 \vartheta^2 + 5 \vartheta - 5}{2 (\vartheta + 1)} \norm{\bfp_C(r)} + (3 \vartheta - 1) U
     + \paren{1 - \frac 1 \vartheta} (T(r) + \tau_1(r+1) - \tau_1(r))\,.
  \end{equation*}
\end{cor}
\begin{proof}
Algorithm~\ref{alg:cluster-sync} differs from the one in~\cite{khanchandani18self} (for hardware clocks with rate $\hnom_v$) in that it does not instantaneously adjust its logical clocks.  Thus, logical clocks may differ in the third phase of the two algorithms. However, Lemma~\ref{lem:nominal-round-length} shows that the result at the end of each round is identical, and neither algorithm takes actions or reacts to messages received in its third phase. Accordingly, the result transfers to Algorithm~\ref{alg:cluster-sync}.
\end{proof}
While Corollary~\ref{cor:pulse-skew-orig} assumes that hardware clock rates are in the range $[1, \vartheta]$, we will later use that under certain conditions, nominal clock rates satisfy more precise bounds than $\hnom_v(t)\in [1,\vartheta_g]$. Using Lemma~\ref{lem:speedup}, we obtain a stronger bound in this setting.

\begin{cor}
  \label{cor:pulse-skew}
  Fix an execution $X$ of Algorithm~\ref{alg:cluster-sync} on a cluster $C$ and a round $r$. Suppose that for each $v \in C$, $v$'s nominal clock satisfies $\zeta \leq \hnom_v(t) \leq \zeta \cdot \vartheta$ for all $t \in [\min \bfp_C(r), p_v(r + 1)]$. Finally suppose that the round $r$ is properly executed for the $\zeta$-reduced execution $\bar{X}$ (cf. Definition~\ref{dfn:sigma-reduced}). Then
  \begin{equation*}
    \norm{\bfp_C(r + 1)} \leq \frac{2 \vartheta^2 + 5 \vartheta - 5}{2 (\vartheta + 1)} \norm{\bfp_C(r)} + (3 \vartheta - 1) U
    + \frac{1}{\zeta} \paren{1 - \frac 1 \vartheta} (T(r) + \tau_1(r+1) - \tau_1(r))\,.
  \end{equation*}
\end{cor}
\begin{proof}
  By the conclusion of Lemma~\ref{lem:speedup}, the nominal clocks satisfy $1 \leq \bar{\hnom_v}(t) \leq \vartheta$ in $\bar{X}$. Then for this execution, the $\ClusterSync$ algorithm guarantees the bound of Corollary~\ref{cor:pulse-skew-orig}, namely
  \begin{equation}
    \label{eqn:speedup-pulse}
    \norm{\bar{\bfp}_C(r + 1)} \leq \frac{2 \vartheta^2 + 5 \vartheta - 5}{2 (\vartheta + 1)} \norm{\bar{\bfp}_C(r)} + (3 \vartheta - 1) \bar{U}
    + \paren{1 - \frac 1 \vartheta} (T(r) + \tau_1(r+1) - \tau_1(r))\,,
  \end{equation}
  where $\bar{\bfp}_C(r)$ is the round $r$ pulse diameter in the transformed execution, and $\bar{U} = \zeta U$. Since (by the proof of Lemma~\ref{lem:speedup}) any event occurring at time $t$ in $X$ has a corresponding event at time $\zeta t$ in $\bar{X}$, we have
  \[
  \norm{\bar{\bfp}_C(r)} = \zeta \norm{\bfp_C(r)} \quad\text{and}\quad \norm{\bar{\bfp}_C(r+1)} = \zeta \norm{\bfp_C(r+1)}.
  \]
  Plugging these expressions into~(\ref{eqn:speedup-pulse}) and dividing by $\zeta$ gives the desired result. 
\end{proof}

%###
\subsection{Sufficient conditions for proper execution.}
%###
In this section, we provide sufficient conditions for a sequence of rounds to be executed properly in a cluster $C$. Given an upper bound $e(1)$ on the diameter of (Newtonian) times at which the nodes in $C$ make their first pulse, we define a sequence $e(2), e(3), \ldots$ of errors recursively (using the bound of Lemma~\ref{cor:pulse-skew-orig}) such that $e(r)$ gives an upper bound on the diameter of the pulse times in round $r$. Using this bound, we define parameters $\tau_1(r), \tau_2(r)$ and $\tau_3(r)$ such that all rounds are properly executed. Thus an upper bound on the skew within $C$ follows from Lemma~\ref{lem:cluster-clock-skew}.

% Suppose $e(r)$ is an upper bound on the cluster $C$ skew at the beginning of round $r$ and all nominal clocks run at rates $\zeta \leq \hnom(t) \leq \zeta \cdot \vartheta$. Suppose nodes $v$ and $w$ begin phase~1 of round~$r$ at (Newtonian) times $t_v(r)$ and $t_w(r)$, respectively with $t_v < t_w$. Since $t_w - t_v \leq e(r)$, $t_v$'s nominal clock could not have increased by more than $\zeta \cdot \vartheta e(r)$ at time $t_w$. Similarly, if $v$ and $w$ pulse at times $p_v(r)$ and $p_w(r)$ (respectively), then $v$ must receive $w$'s pulse within nominal time $\zeta \vartheta (e(r) + d)$. Finally, (by some lemma) the correction computed by $v$ at the end of round $r$, $\Delta_v(r)$, satisfies $\abs{\Delta_v(r)} \leq e(r) + U$. Thus, a local adjustment of $\Delta_v(r)$ can be made to $v$'s logical clock by speeding or slowing the clock rate by up to $\varphi$ for a duration of at most $(e(r) + U) / \varphi$. In accordance with these observations, we make the following definition.

\begin{dfn}
  \label{dfn:feasible}
  We say that round $r$ is \dft{feasible} if the following conditions hold:
  \begin{equation}
    \label{eqn:taus}
    \begin{split}
      \tau_1(r) &\geq \thg \cdot e(r)\\
      \tau_2(r) &\geq \thg \cdot (e(r) + d)\\
      \tau_3(r) &\geq \thg \cdot \frac{1}{\varphi} \cdot (e(r) + U),
    \end{split}
  \end{equation}
  where $\thg = (1 + \rho)(1 + \mu)$.
\end{dfn}

\begin{cor}[of~\cite{khanchandani18self}, Lemma~6]
  \label{cor:feasible-implies-proper}
  If in an execution of Algorithm~\ref{alg:cluster-sync} it holds that $\norm{\bfp(r)}\leq e(r)$ and round $r$ is feasible, then round $r$ is properly executed.
\end{cor}
\begin{proof}
As nominal clock rates are from $[1,\vartheta_g]$ and $0<\varphi<1$, feasibility of round $r$ and the assumption that $\bfp(r)\leq e(r)$ implies the preconditions of Lemma~6 from~\cite{khanchandani18self}. The lemma directly implies the first two requirements of proper execution. Its proof also establishes that for each $v\in C$, $\Delta_v(r)\leq \vartheta_g(\|\bfp(r)\|+U)\leq \vartheta_g(e(r)+U)\leq \varphi\tau_3(r)$, showing the third property.
\end{proof}

The following result gives an inductive formula for computing an upper bound $e(r+1)$ on the cluster skew given an upper bound $e(r)$.

\begin{cor}
  \label{cor:inductive-error-bound}
  Suppose $\norm{\bfp_C(r)} \leq e(r)$ and $\tau_1(r), \tau_2(r)$ and $\tau_3(r)$ are defined by taking equalities in Equations~(\ref{eqn:taus}). Then for $e(r + 1)$ satisfying
  \begin{equation}
    \label{eqn:pre-inductive-error-bound}
    e(r + 1) \geq \frac{2 \vartheta_g^2 + 5 \vartheta_g - 5}{2 (\vartheta_g + 1)} e(r) + (3 \vartheta_g - 1) U
    + \paren{1 - \frac{1}{\vartheta_g}} (T(r) + \tau_1(r+1) - \tau_1(r))\,,
  \end{equation}
  we have $\norm{\bfp_K(r+1)} \leq e(r+1)$.
\end{cor}
\begin{proof}
  Readily follows from Corollaries~\ref{cor:feasible-implies-proper} and~\ref{cor:pulse-skew-orig} with $\vartheta=\vartheta_g$.
\end{proof}

The expression~(\ref{eqn:pre-inductive-error-bound}) is somewhat inconvenient because term $\tau_1(r+1)$ on the right side generally depends on $e(r+1)$. One can eliminate the dependence of the parameters on $r$ by fixing $e(r)=E$ for all $r\in \N$, which is feasible if and only if $E\geq e(1)$ and Inequality~\ref{eqn:pre-inductive-error-bound} holds when replacing $e(r+1)$ and $e(r)$ by $E$. We can find the minimal such $E$ by solving for the value of $E$ for which equality is attained and taking the larger of this value and the initial skew bound $e(1)$. Given that we assume perfect initialization (i.e., all logical clocks are initialized to $0$ at time $0$), this will always result in the value achieving equality in Inequality~\ref{eqn:pre-inductive-error-bound}. Note that $\tau_1(r)$, $\tau_2(r)$, $\tau_3(r)$, and $T(r)$ then all become independent of $r$, i.e.,
  \begin{equation}
    \label{eqn:taus_fixed}
    \begin{split}
      \tau_1 &= \vartheta_g \cdot E\\
      \tau_2 &= \vartheta_g \cdot (E + d)\\
      \tau_3 &= \vartheta_g \cdot \frac{1}{\varphi} \cdot (E + U)\\
      T &= \left(2+\frac{1}{\varphi}\right)\vartheta_g\cdot E + \vartheta_g \cdot d +\vartheta_g \cdot \frac{U}{\varphi}\,.
    \end{split}
  \end{equation}

However, the above procedure is only feasible if such a solution for \eqref{eqn:pre-inductive-error-bound} exists. Plugging the above values for $T$ and $\tau_1$ into \eqref{eqn:pre-inductive-error-bound}, the expression simplifies to $E = \alpha E + \beta$, where
\begin{equation}
  \label{eqn:alpha-beta}
  \begin{split}
    \alpha &= \frac{6 \vartheta^2_g \varphi + 5 \vartheta_g \varphi - 9 \varphi + 2 \vartheta_g^2 - 2}{2 \varphi (\vartheta_g + 1)}\\
    \beta &= \left(3 \vartheta_g - 1 + \frac{\vartheta_g-1}{\varphi}\right) U + (\vartheta_g - 1) d\,.
  \end{split}
\end{equation}
This equality can be solved if and only if $\alpha<1$, yielding $E = \frac{\beta}{1-\alpha}$. Observe that the $\varphi$ contributes as an additive $\Theta((\vartheta_g-1)/\varphi)$ term to $\alpha$. This indicates that we can choose $\varphi\in \Theta(1/(\vartheta_g-1))$, which we will leverage to ensure that the logical clock drift of the algorithm will be $O(\rho)$, i.e., logical clocks will not behave much worse than free-running hardware clocks, despite the achieved synchronization guarantees.

Putting it all together, we obtain the following general result.

\begin{prop}
  \label{prop:cluster-error-bound}
 Fix a cluster $C\in \mathcal{C}$, suppose that $\alpha<1$ in \eqref{eqn:alpha-beta}, let $E=\frac{\beta}{1-\alpha}$, and choose $\tau_i(r)=\tau_i$ for $i\in \{1,2,3\}$ and all $r\in \N$, where $\tau_i$ is given by \eqref{eqn:taus_fixed}. If $\norm{\bfp(1)}\leq E$ when executing Algorithm~\ref{alg:cluster-sync}, then every round $r\in \N$ is properly executed and $\norm{\bfp(r)}\leq E$ for all $r\in \N$.
\end{prop}
\begin{proof}[Proof of Corollary~\ref{cor:cluster-skew-bound}]
We apply Proposition~\ref{prop:cluster-error-bound}, Lemma~\ref{lem:cluster-clock-skew}, and use that $\vartheta_g E$ upper bounds $(\vartheta_g - 1) T$, as it satisfies \eqref{eqn:pre-inductive-error-bound} with equality.
\end{proof}

Lemma~\ref{lem:unanimous-error-gap} follows from several claims proven below.

\begin{claim}
  \label{claim:gen-error-bound}
  Consider the sequences $\tau_1, \tau_2$ and $\tau_3$ defined in~(\ref{eqn:taus-c}), and the sequence $e(r)$ is defined by $e(r + 1) = \alpha \cdot e(r) + \beta$ where
  \begin{equation}
    \label{eqn:abg}
    \begin{split}
    \alpha &= \frac{2 \vartheta^2 + 5 \vartheta - 5}{2 (\vartheta + 1) (1 - \gamma)} + \frac{\gamma}{1 - \gamma} (1 + c_1)\\
    \beta &= \frac{\gamma}{1 - \gamma} d + \frac{1}{1 - \gamma} ((3 \vartheta - 1) + \gamma \cdot c_1) U\\
    \gamma &= \frac{\zeta_\mx}{\zeta} \cdot \frac{\vartheta_g}{\vartheta} \cdot (\vartheta - 1).      
    \end{split}
  \end{equation}
  Let $\alpha_g$ ($\beta_g$), $\alpha_f$ ($\beta_f$), $\alpha_s$ ($\beta_s$) denote the values of $\alpha$ ($\beta$) when a cluster is in the general, unanimously fast, and unanimously slow state, respectively. Then
  \begin{align*}
    \alpha_g &= \frac 1 2 + (1 + c_2 ) c_1 \cdot \rho + O(\rho + c_1 \cdot \rho^2)\\
    \beta_g &= (1 + c_2 + O(\rho)) \rho \cdot d + (2 + (1 + c_2) c_1 \cdot \rho + O(c_1 \cdot \rho^2)) U\\
    \alpha_f &= \frac 1 2 + c_1 \cdot \rho + O(\rho + c_1 \cdot \rho^2)\\
    \beta_f &= (\rho + O(\rho^2)) d + (2 + c_1 \rho + O(\rho)) U.\\
    \alpha_s &= \frac 1 2 + c_1 \cdot \rho + O(\rho + c_1 \cdot \rho^2)\\
    \beta_s &= (\rho + O(\rho^2)) d + (2 + c_1 \rho + O(\rho)) U.
  \end{align*}
\end{claim}
\begin{proof}
  All of the expressions above follow from~(\ref{eqn:abg}) with the appropriate setting of parameters, and using the fact that $1 / (1 - x) = 1 + x + O(x^2)$.
\end{proof}

\begin{claim}
  \label{claim:alpha-ratio}
  For any constants $c_2 \geq 1$ and $\e > 0$, and for sufficiently small $\rho > 0$ we have
  \[
  c_1 = \frac{\frac 1 2 - \e}{1 + c_2} \cdot \frac 1 \rho \implies \frac{1 - \alpha_g}{1 - \alpha_f} \leq 4 \e.
  \]
\end{claim}
\begin{proof}
  By plugging in the value of $c_1$ above into the expression for $\alpha_g$ in Claim~\ref{claim:gen-error-bound}, we get $1 - \alpha_g = \e + O(\rho)$. Therefore,
  \begin{align*}
    \frac{1 - \alpha_g}{1 - \alpha_f} &= \frac{\e + O(\rho)}{\frac 1 2 - ((1/2) - \e)(1 + c_1) - O(\rho)}\\
    &= \frac{(1 + c_1) \e + O(\rho)}{(1/2) (1 + c_1) - ((1/2) - \e - O(\rho)}\\
    &= \frac{1 + c_1}{(c_1 / 2) + \e}\e + O(\rho)\\
    &\leq 2 \frac{1 + c_1}{c_1} \e\\
    &\leq 4 \e.
  \end{align*}
\end{proof}

\begin{claim}
  \label{claim:steady-state-ratio}
  For any constant $c_2 \geq 1$ and sufficiently small $\rho > 0$, setting $c_1$ as in Claim~\ref{claim:alpha-ratio} with $\e = 1 / 4096$ we obtain
  \[
  \frac{\essf}{\essg} \leq \frac{c_1 \cdot c_2}{128} \rho.
  \]
\end{claim}
\begin{proof}
  Recall that $\essf = \beta_f / (1 - \alpha_f)$, and similarly for $\essg$. By Claim~\ref{claim:gen-error-bound}, we have $\beta_f / \beta_g \leq 1$ for all $c_1 \geq 1$ and sufficiently small $\rho$. Thus, it suffices to show that
  \begin{equation}
    \label{eqn:alpha-ratio}
    \frac{1 - \alpha_g}{1 - \alpha_f} \leq \frac{c_1 \cdot c_2}{128} \rho.    
  \end{equation}
  Plugging in $c_1$ as in Claim~\ref{claim:alpha-ratio}, we compute
  \[
  \frac{c_1 \cdot c_2}{128} \rho \geq \frac{1}{512} \frac{c_2}{1 + c_2} \geq \frac{1}{1024},
  \]
  where the first inequality holds for any $\e \leq 1/4$. Thus, it suffices to show that $(1 - \alpha_g) / (1 - \alpha_f) \leq 1 / 1024$. By Claim~\ref{claim:alpha-ratio}, setting $\e = 1 / 4096$ suffices.
\end{proof}

\begin{proof}[Proof of Lemma~\ref{lem:unanimous-error-gap}, Part~1]
  Here, we show that.
  \begin{equation}
    \label{eqn:u-error-gap}
    \frac{L_v(t_v(r + 1)) - L_v(t_v(r))}{t_v(r + 1) - t_v(r)} \geq (1 + \varphi)\paren{1 + \frac 7 8 \mu}.
  \end{equation}
  For a clock $v$ in fast mode, $\hnom_v(t) \geq (1 + \varphi)(1 + \mu)$. The nominal length of a round $r$ for $v$ is at most the logical round length, plus the Lynch-Welsh correction: $T(r) + \Delta_v(r)$. Therefore, we have
  \[
  t_v(r + 1) - t_v(r) \leq \frac{T(r) + \Delta_v(r)}{(1 + \varphi)(1 + \mu)}. 
  \]
  Thus, it suffices to prove that
  \[
  T(r) \paren{\frac{T(r) + \Delta_v(r)}{(1 + \varphi)(1 + \mu)}}^{-1} \geq (1 + \varphi)\paren{1 + \frac 7 8 \mu}. 
  \]
  Rearranging this expression gives the equivalent expression
  \[
  \frac{\Delta_v(r)}{T(r) + \Delta_v(r)} \leq \frac{\mu}{8(1 + \mu)}.
  \]
  Observing that $\Delta_v(r) \leq \frac 1 2 T(r)$ and assuming $1 + \mu \leq 2$, it suffices to show that
  \begin{equation}
    \label{eqn:delta-bound}
    \Delta_v(r) \leq \frac{\mu}{32} T(r).
  \end{equation}
  We can bound the correction $\Delta_v(r)$---assuming that $e_{f, k}(r)$ is sufficiently close to $\essf$\footnote{We will argue that $k = O(1)$ suffices later.}---as follows:
  \[
  \Delta_v(r) \leq \underset{\leq 2}{\underbrace{(1 + \varphi)(1 + \mu)(1 + \rho)}} \cdot \underset{\leq 2 \essf}{\underbrace{(\norm{\bfp(r)} + U)}} \leq 4 \essf.
  \]
  Again, the bound on the second term holds for sufficiently large $k$. Since $T(r) = \tau_1(r) + \tau_2(r) + \tau_3(r)$ are taken as in~(\ref{eqn:taus-c}), we have $T(r) \geq c_1 \essg$. Further, by assumption, we take $\mu = c_2 \cdot \rho$.  Thus to show~(\ref{eqn:delta-bound}) it suffices to show that
  \[
  \frac{\essf}{\essg} \leq \frac{c_1 \cdot c_2}{128} \rho.
  \]
  By Claim~\ref{claim:steady-state-ratio}, this final equation is satisfied for any $c_2 \geq 1$ and $c_1 = (1/2 - \e) / (1 + c_2) \rho$ for $\e = 1/4096$. Thus~(\ref{eqn:u-error-gap}) holds for every node $v \in C \setminus F$ individually. Since $F_C(t)$ is defined to be the median value $\set{L_v(t) \sucht v \in C \setminus F}$, $F_C(t)$ must also increase at (at least) the minimum rate of the individual $L_v(t)$s (assuming, of course, that $F$ is static). 

  All that remains is to bound the number of rounds, $k$, until $\norm{\bfp(r)} + U \leq 2_f^\infty$. To this end, we apply Claims~\ref{claim:gen-error-bound} with the setting of parameters as in Claims~\ref{claim:alpha-ratio} and~\ref{claim:steady-state-ratio}. Recall that the steady state errors are defined by $\essf = \beta_f / (1 - \alpha_f)$ and $\essg = \beta_g / (1 - \alpha_g)$. From the choice of $c_1$ (Claim~\ref{claim:alpha-ratio}), we get
  \[
  1 - \alpha_g = \e - O(\rho).
  \]
  Thus, for sufficiently small $\rho$, we have
  \[
  \essg = \frac{\beta_g}{1 - \alpha_g} = \frac{2}{\e} \rho (1 + c_2 + O(\rho)) \cdot d + \frac{2}{\e}(2 + \e + O(\rho)) \cdot U.
  \]
  Similarly, 
  \[
  1 - \alpha_f = \frac 1 2 - \rho \cdot c_1 - O(\rho) = \frac 1 2 - \frac{\frac 1 2 - \e}{1 + c_2} - O(\rho) \geq 1/4.
  \]
  Therefore
  \[
  \essf = \frac{\beta_g}{1 - \alpha_g} \leq 4 \rho (1 + c_2 + O(\rho)) \cdot d + 4 (2 + \e + O(\rho)) \cdot U.
  \]
  Thus, $\essg / \essf = \Theta(1)$. Since $e(r-k) \leq 2 \essg$ by assumption and the convergence of the error to $\essf$ is exponential, the number of rounds until the error is at most $c \cdot \essf$ for any constant $c > 1$ is constant. Thus $k = O(1)$ suffices.
\end{proof}

\begin{proof}[Proof of Lemma~\ref{lem:unanimous-error-gap}, Part~2]
  Here we show:
  \[
  (1 + \varphi)\paren{1 - \frac 1 8 \mu} \underset{\text{(a)}}{\leq} \frac{L_v(t_v(r + 1) - L_v(t_v(r))}{t_v(r + 1) - t_v(r)} \underset{\text{(b)}}{\leq} (1 + \varphi)\paren{1 + \frac 1 8 \mu}
  \]
  We first consider inequality~(a). Arguing as in the proof of Part~1, it suffices to show that $\abs{\Delta_v(r)} \leq \frac 1 8 \mu T(r)$. As before, we can bound $\abs{\Delta_v(r)} \leq 4 \esss$, so it suffices to have $\esss \leq \frac{1}{32} \mu T(r)$. Since $T(r) \geq c_1 \essg$, it suffices to have
  \begin{equation}
    \label{eqn:esss-sufficiency-a}
    \frac{\esss}{\essg} \leq \frac{c_1 \cdot c_2}{32} \rho.
  \end{equation}
  Before proving that~(\ref{eqn:esss-sufficiency-a}) is satified, we will show that inequality~(b) requires a more stringent condition (which is also satisfied for appropriate choices of parameters).

  For (b), since each $v \in C \setminus F$ runs in slow mode, we have
  \[
  \frac{T(r) - \abs{\Delta_v(r)}}{(1 + \varphi)(1 + \rho)} \leq t_v(r+1) - t_v(r).
  \]
  Rearranging terms gives
  \[
  \frac{T(r)}{\frac{T(r) - \abs{\Delta_v(r)}}{(1 + \varphi)(1 + \rho)}} \leq (1 + \varphi)\paren{1 + \frac 1 8 \mu} \iff \frac{\Delta_v(r) + \rho T(r)}{T(r) - \abs{\Delta_v(r)}} \leq 1 - \frac 1 8 \mu.
  \]
  Using the fact that $\abs{\Delta_v(r)} \leq T(r) / 2$, it is sufficient to show that
  \[
  \frac{\abs{\Delta_v(r)}}{T(r)} + \rho \leq \frac 1 {16} \mu.
  \]
  Taking $c_2 \geq 32$---i.e., $\rho \leq (1/32) \mu$---we find that it suffices to show
  \[
  \frac{\abs{\Delta_v(r)}}{T(r)} \leq \frac 1 {32} \mu.
  \]
  Thus, in this case we require
  \begin{equation}
    \label{eqn:esss-sufficiency-b}
    \frac{\esss}{\essg} \leq \frac{c_1 \cdot c_2}{128} \rho \quad\text{and}\quad c_2 \geq 32.
  \end{equation}
  Since condition~(eqn:sss-sufficiency-b) implies~(eqn:sss-sufficiency-a), we must only prove the latter. The argument is nearly identical to the one for Part~1 of the lemma.
\end{proof}

\section{Analysis of the Intercluster Algorithm}\label{app:gcs}

\subsection{Missing proofs from Section~\ref{sec:intercluster-alg}}

\begin{proof}[Proof of Lemma~\ref{lem:faithful-delta}]
  We consider the case where $C$ satisfies $\FC$ at time $t$. The case where $C$ satisfies $\SC$ is analogous. We must show that for all $v \in C \setminus F$, $v$ satisfies $\FT$ at every time in the interval $[t_v(r_t - k), t_v(r_t)]$. First consider the case where $C$ satisfies $\FCone$ at time $t$. That is, there exists a neighboring cluster $A$ of $C$ and an integer $s$ such that $L_A(t) - L_C(t) \geq 2 s \kappa$. Suppose towards a contradiction that $v \in C \setminus F$ does not satisfy $\FTone$ at time $t' = t_v(r_t - k)$ (for the same value of $s$ for which $C$ satisfies $\FCone$ at time $t$). That is,
  \begin{equation}
    \label{eqn:not-ftone}
    \tilL_A^v(t') - L_v(t') < 2 s \kappa - \delta.
  \end{equation}
  Then we infer:
  \begin{align*}
    (\ref{eqn:not-ftone}) &\implies \tilL_A^v(t') - L_A(t') + L_A(t') - L_v(t') < 2 s \kappa - \delta\\
    &\implies - 2 \Err +  L_A(t') - L_v(t') < 2 s \kappa - \delta\\
    &\implies L_A(t') - L_v(t') < 2 s \kappa - \delta + 2 \Err\\
    &\implies L_A(t') - L_A(t) + L_C(t) - L_v(t') < - \delta + 2 \Err\\
    &\implies \delta < 2 \Err + (L_A(t) - L_A(t')) - (L_C(t) - L_C(t')) + (L_v(t') - L_C(t'))\\
    &\implies \delta < 2 \Err + \thmax (t - t') - (t - t') + 2 \Err\\
    &\implies \delta < 4 \Err + (\thmax - 1)(t - t')\\
    &\implies \delta < 4 \Err + (\thmax - 1)(k + 1) T\\
    &\implies \delta < 4 \Err + (k + 1) \Err = (k + 5) \Err.
  \end{align*}
  The final expression contradicts the choice of $\delta$. Thus every $v$ satisfies $\FTone$ at the beginning of round $r_t - k$. A nearly identical argument shows that each $v$ satisfies $\FTtwo$ at the beginning of round $r_t - k$ as well. Finally, since $\kappa = 3 \delta$, $\FT$ and $\ST$ are mutually exclusive. Therefore, every execution of $X$ is faithful for $C$, as desired.
\end{proof}

\begin{proof}[Proof of Proposition~\ref{prop:is-simulates-gcs}]
  Fix a cluster $C \in \calC$, we show step by step that all four axioms are satisfied starting with axiom (A1).
  From definition of $L_C$ we get $1 \leq \frac{\mathrm{d}}{\mathrm{dt}} L_C(t)$ for free. Hence, we need to
  show that $\frac{\mathrm{d}}{\mathrm{dt}} L_C(t) \leq (1 + \bar{\rho})(1 + \bar{\mu})$.
  We examine both sides of the inequality starting with the LHS.
  \begin{align*}
    \frac{\mathrm{d}}{\mathrm{dt}} L_C(t) & \leq (1+\frac{2\varphi}{1-\varphi})(1+\rho)(1+\mu)\\
    &= (1+2\varphi+O(\rho^2))(1+\rho)(1+\mu)\\
    &= 1 + 2\varphi + \rho + \mu + O(\rho^2)\\
    &= 1 + 2\varphi + (c_2 + 1)\rho + O(\rho^2)
  \end{align*}
  The RHS gives us:
  \begin{align*}
    (1 + \bar{\rho})(1 + \bar{\mu}) &= (1+\varphi)^2(1+\frac{1}{4}\mu)(1+\frac{7}{8}\mu)\\
    &=(1+2\varphi+O(\rho^2))(1+\frac{9}{8}\mu+O(\rho^2))\\
    &= 1 + 2\varphi + \frac{9}{8}c_2\rho + O(\rho^2)
  \end{align*}
  Putting both sides together the first axiom is satisfied  for sufficiently small $\rho$, when
  $$1 + 2\varphi + (c_2 + 1)\rho + \rho \leq 1 + 2\varphi + \frac{9}{8}c_2\rho\,,$$
  i.e., if $c_2 \geq 16$. Which is already implied by condition~(\ref{eqn:esss-sufficiency-b}).

  By Corollary~\ref{cor:faithful-error-gap} we know that if $C$ satisfies $\SC$
  at time $t$, then $\frac{\mathrm{d}}{\mathrm{dt}} L_v(t)$ is bounded above by
  $(1+\varphi)(1+\frac{1}{8}\mu)$. Similarly, if $C$ satisfies $\FC$ at time
  $t$, then $\frac{\mathrm{d}}{\mathrm{dt}} L_v(t)$ is bounded below by
  $(1+\varphi)(1+\frac{7}{8}\mu)$. Which gives us that the following two axioms,
  i.e., axiom (A2) $ (1+\varphi)(1+\frac{1}{8}\mu) \leq 1+\bar{\rho}$ and axiom
  (A3) $1+\bar{\mu} \leq (1+\varphi)(1+\frac{7}{8}\mu)$ are readily satisfied by
  the choice of $\bar{\rho}$ and $\bar{\mu}$.

  We are left with proving the fourth axiom (A4) $\bar{\mu} / \bar{\rho} > 1$, for sufficiently small $\rho$.
  \begin{align*}
    \bar{\mu}/\bar{\rho}=\frac{(1 + \varphi)(1 + (7/8)\mu) - 1}{(1 + \varphi)(1 + (1/4) \mu) - 1} &=
    \frac{\varphi + (7/8)\mu + O(\rho^2)}{\varphi + (1/4)\mu + O(\rho^2)} \\
    &\geq \frac{\varphi + (7/8)\mu + \rho}{\varphi + (1/4)\mu} \\
    &= 1 + \frac{(5/8)\mu + \rho}{\varphi + (1/4)\mu}
  \end{align*}
  Hence, $\bar{\mu}/\bar{\rho} \geq 1+\eta$ for some $\eta\leq\frac{(5/8)\mu + \rho}{\varphi + (1/4)\mu}$
  and thus axiom (A4) is satisfied.
\end{proof}

\subsection{Bounding the global skew}

So far, we have neglected the global skew. Bounding it provides, essentially, the induction base for proving a small local skew. As no new techniques are required for ensuring a global skew of $O(\delta D)$, where $D$ is the hop diameter of both $G$ and $\cal{G}$, we just sketch a feasible construction here. A standard approach from the literature~(see, e.g., \cite{kuhn18gradient,lenzen10tight}) is to ensure that (i) all nodes maintain a conservative estimate of the maximum clock value in the system that is at most by the target bound $\mathcal{S}$ behind the actual value, (ii) nodes attaining the maximum increase their clocks slowly, and (iii) nodes whose clocks are by the target bound behind the maximum are guaranteed to run faster than those to which (ii) applies. From this is immediate that the target bound can never be surpassed. Naturally, the rules ensuring these properties must not come into conflict with the fast and slow mode triggers, as the algorithm must still obey them.

We sketch how to simultaneously achieve these properties for our setting and an asymptotically optimal skew bound of $O(\delta D)$. We start by defining $L^{\max}(t):=\max_{v\in V\setminus F}\{L_v(t)\}$ and bounding the rate at which it increases.
\begin{lem}\label{lem:Lmax}
Suppose nodes pick slow mode whenever the fast mode trigger is not satisfied at the beginning of a round. Then $L^{\max}$ increases at rate at least $1$ and at most $1+O(\rho)$.
\end{lem}
\begin{proof}[Proof Sketch.]
The lower bound is immediate from Lemma~\ref{lem:logical-clock-rate}.

Consider some correct node $v\in C$ and round $r$. We make a case distinction. If $v$ is in fast mode in round $r$, then by assumption $v$ satisfied the fast trigger and there was some correct neighbor $w$ of $v$ whose logical clock was by at least $2\kappa-\delta > \delta$ than that of $v$ at time $t_v(r)$. However, by the choice of $\delta \geq 4 \Err \geq (\thmax - 1) T + \thmax \cdot \Err$, this implies that $L_w(t) > L_v(t)$, so that $L_v(t) \leq L^{\max}(t)$ for all $t \in [t_v(r), t_v(r+1)]$.

Now suppose $v$ is in slow mode during round $r$. Consider in Algorithm~\ref{alg:cluster-sync} the values $L_v(t_{wv})$ stored for node $w\in C$ in some round $r$ at time $t_{wv}$ due to receiving a message from $w$. The message was sent at logical time $(r-1)T+\tau_1$; denote by $t^{\max}$ the time when the first correct clock reached this value. Thus, $t_{wv}\geq t^{\max}+d-U$ for all $w\in C$. Node $v$ compares its logical clock value at reception times of these messages to the value for its own message, which satisfies $t_{vv}\leq p_v(r)+d$, where $p_v(r)$ is the time when $v$ broadcast its own pulse in round $r$. Note that the logical clock of nodes run at least at rate $1+\varphi$ between sending their pulses and them being received, as they are in phase 2 during this time. Moreover, as $v$ is in slow mode, it $L_v$ runs at rate $(1+\varphi)h_v(\tau)\leq (1+\varphi)(1+\rho)$ for $\tau\in [p_v(r),t_{vv}]$. Therefore, we get that
\begin{align*}
L_v(t_{wv})-L_v(t_{vv})&\geq L_v(t^{\max}+d-U)-L_v(p_v(r)+d)\\
&\geq L_v(t^{\max})+(1+\varphi)(d-U)-L_v(p_v(r))-(1+\varphi)(1+\rho)d\\
&\geq L_v(t^{\max})-L^{\max}(t^{\max}) - U - O(\rho d)\,.
\end{align*}
This bounds $\Delta_v$ such that $v$ will not increase its logical clock by more than this value compared to its nominal rate. We get that $L_v(t)\leq L^{\max}(t_r(v))+O(\rho)(t-t_v(r))+U$ for all $t\in [t_v(r),t_v(r+1)]$.
\end{proof}
Next, we discuss how to maintain the desired estimate of $L^{\max}$. For simplicity, we make no attempt to keep the message complexity low.
\begin{lem}\label{lem:estimate_global}
Suppose nodes pick slow mode whenever the fast mode trigger is not satisfied at the beginning of a round. Then we can maintain estimates $M_v(t)$ such that $L^{\max}(t)\geq M_v(t)\in L^{\max}(t)-O(\delta D)$.
\end{lem}
\begin{proof}[Proof sketch.]
Each node initializes $M_v(0):=0$ and increases $M_v$ at rate $h_v(t)/(1+\rho)\leq 1$. By Lemma~\ref{lem:Lmax} this can never cause $M_v$ to exceed $L^{\max}$. In addition, nodes transmit pulses (distinguishable from the ones for providing their actual clock values to neighbors) whenever $M_v(t)$ reaches a multiple of $d-U$. Receivers memorize these pulses. When a node registered $k$ such pulses each from $f+1$ nodes in any adjacent cluster, it sets $M_v(t)$ to $(k+1)(d-U)$, provided this increases $M_v$, and sends out the respective pulses. As there are at most $f$ faults in each cluster and messages are under way for at least $d-U$ time, it is still guaranteed that $M_v(t)\leq L^{\max}(t)$.

From the upper bound provided by Lemma~\ref{lem:Lmax}, it follows that $L^{\max}(t)\leq M_v(t)+O(\rho t)$ for all $t$, showing the claim for all times $t\in O(\delta D/\rho)\subseteq O(d D)$. Hence, assume that $t\geq (d+1)D$. Choose $v$ such that $L_v(t-d(D+1))=L_v^{\max}(t-d(D+1))=:L$. As clusters are synchronized up to uncertainty $\delta$, all correct nodes $w$ in $v$'s cluster $C$ satisfy that
\begin{equation*}
M_w(t-d(D+1))\geq L_w(t-d(D+1))\geq L-\delta\,.
\end{equation*}
Note that within $d-U$ time, these nodes will send a pulse indicating some threshold $k(d-U)$ for their $M_w$ values, which increase at least at rate $1-\rho$. These messages are under way for at most $d$ time, and the recipients will set their estimates to at least $(k+1)(d-U)$ upon receiving these values, causing them to send pulses themselves (if they did not do so yet), and so on. Note that this always holds for all nodes in a cluster, hence the result is a fault-tolerant flooding. By time $t$, each correct node $w$ in the graph will have an estimate of $M_w(t)\in L+(d-U)(D+1)-O(\rho d D)$: Nodes locally increase estimates at rate $1-\rho$, while estimates that are communicated are under way for at most $d$ time and cause, in essence, their recipients to add $d-U$ to the received value. Using the bound of $1+O(\rho)$ on the rate at which $L^{\max}$ increases one final time, we conclude that for all $w\in V\setminus F$, it holds that
\begin{equation*}
L^{\max}(t)\in L+(1+O(\rho))d(D+1)\subseteq M_w(t) - U(D+1)-O(\rho d D)\subseteq M_w(t)-O(\delta D)\,.\qedhere
\end{equation*}
\end{proof}
Finally, we need to ensure that nodes that are far behind can indeed use the information gleaned from $M_v$ suffices to catch up when they threaten to fall to far behind.
\begin{thm}\label{thm:global_skew}
Suppose nodes maintain estimates as in Lemma~\ref{lem:estimate_global} choose modes as follows:
\begin{itemize}
  \item They follow the fast and slow mode triggers of Algorithm~\ref{alg:intercluster}.
  \item If neither holds and $L_v(t_v(r))\leq M_v(t_v(r))-c\delta$ (for a sufficiently large constant $c$) at the beginning of round $r$, $v$ switches to fast mode for this round.
  \item If neither of the above triggers holds at time $t_v(r)$, $v$ chooses slow mode in round $r$.
\end{itemize}
Then the global skew is bounded by $O(\delta D)$.
\end{thm}
\begin{proof}[Proof Sketch.]
First, we argue that the second rule does not invalidate Lemmas~\ref{lem:Lmax} and~\ref{lem:estimate_global}. To see this, note that so long as $M_v(t_v(r))\leq L^{\max}(t_v(r))$, the second rule being applied implies for $t\in [t_v(r),t_v(r+1)]$ by Lemma~\ref{lem:logical-clock-rate} that
\begin{align*}
L^{\max}(t)-L_v(t)&\geq L^{\max}(t_v(r))-L_v(t_v(r)+(1-\vartheta_g)(t-t_v(r))\\
&\in L^{\max}(t_v(r))-(M_v(t_v(r)-c\delta) - O(\rho T)\\
&\subseteq (c-O(1))\delta >0\,,
\end{align*}
where the last step uses that $c$ is sufficiently large. In other words, we have that $L^{\max}$ cannot increase faster than it would without the second rule due to $v$ in round $r$, provided that $M_v(t_v(r))\leq L^{\max}(t_v(r))$. Using this argument whenever the second rule is applied, we can repeat the proofs of the two lemmas simultaneously and inductively over time, showing both their statements still apply.

Now assume that at time $t$, $L_v(t)=\min_{w\in V\setminus F}\{L_w(t)\}<L^{\max}(t)-c\delta(D+1)$. As $c$ is sufficiently large, by Lemma~\ref{lem:estimate_global}, we have that $L_v(t)\leq M_v(t_v(r))-c\delta$. Moreover, the slow mode trigger cannot be satisfied by $v$: this would imply that some $w$ satisfied $\tilL_w^v(t_v(r)) \leq L_v(t_v(r)) - \kappa + \delta \leq L_v(t_v(r)) - 4 \delta$. Since $\abs{\tilL_w^v(t_v(r)) - L_w(t_v(r))} \ll \delta$ and $(\thmax - 1) T \ll \delta$, this contradicts the minimality of $L_v(t)$. Thus the clock satisfying $L_v(t) = L^{\min}(t)$ will be in fast mode, as will any clock satisfying $L_v(t) - L^{\min}(t) \leq \delta / 2$. Since the fastest clock is in slow mode, and the slowest clock runs fast, the global skew cannot increase beyond $O(\delta D)$.
\end{proof}

\end{document}